\documentclass[runningheads]{llncs}

\usepackage[utf8]{inputenc}
\usepackage{amsmath}
\usepackage{amssymb}
\usepackage{graphicx}
\usepackage{tabularx}
\usepackage{array}
\usepackage{pifont}
\usepackage{xspace}
\usepackage{todonotes}
\usepackage{paralist}
\usepackage{thm-restate}
\usepackage{subfigure}
\usepackage[pdfpagelabels,colorlinks,citecolor=blue,linkcolor=blue,urlcolor=blue]{hyperref}
\usepackage[english]{babel}
\usepackage{amsopn}
\usepackage{latexsym}
\usepackage{multirow}
\usepackage{multicol}
\usepackage{booktabs}
\usepackage[capitalise]{cleveref}
\usepackage{color, colortbl}
\usepackage{booktabs}

\spnewtheorem{clm}{Claim}{\bfseries}{\rmfamily}

\graphicspath{{figures/}}

\definecolor{lightcyan}{rgb}{0.88,1,1}
\definecolor{antiquewhite}{rgb}{0.98, 0.92, 0.84}

\newcounter{casecounter}
\newcounter{subcasecounter}
\newcounter{subsubcasecounter}
\makeatletter
\newcommand{\ccase}[2][]{%
	\stepcounter{casecounter}%
	\setcounter{subcasecounter}{0}%
	\protected@write \@auxout {}{\string \newlabel {#2}{{#1\thecasecounter}{\thepage}{#1\thecasecounter}{#2}{}} }%
	\hypertarget{#2}{\noindent\textbf{Case #1\thecasecounter.}}
}

\newcommand{\subcase}[2][]{%
	\stepcounter{subcasecounter}%
	\setcounter{subsubcasecounter}{0}%
	\protected@write \@auxout {}{\string \newlabel {#2}{{#1\thecasecounter.\thesubcasecounter}{\thepage}{#1\thecasecounter.\thesubcasecounter}{#2}{}} }%
	\hypertarget{#2}{\noindent\textbf{Case #1\thecasecounter.\thesubcasecounter.}}
}

\newcommand{\subsubcase}[2][]{%
	\stepcounter{subsubcasecounter}%
	\protected@write \@auxout {}{\string \newlabel {#2}{{#1\thecasecounter.\thesubcasecounter.\thesubsubcasecounter}{\thepage}{#1\thecasecounter.\thesubcasecounter.\thesubsubcasecounter}{#2}{}} }%
	\hypertarget{#2}{\noindent\textbf{Case #1\thecasecounter.\thesubcasecounter.\thesubsubcasecounter.}}
}
\makeatother

\begin{document}
	\title{Mutual Witness Gabriel Drawings \\ of Complete Bipartite Graphs
\thanks {Work partially supported by: (i) MUR, grant 20174LF3T8 AHeAD: efficient Algorithms for HArnessing networked Data", (ii) Dipartimento di Ingegneria, Universita degli Studi di Perugia, grant RICBA21LG: Algoritmi, modelli e sistemi per la rappresentazione visuale di reti.}
}

\author{William J. Lenhart\inst{1},
		Giuseppe Liotta\inst{2}}

\institute{Williams College, US\\
		\email {wlenhart@williams.edu}
		\and
        Universit\`a degli Studi di Perugia, Italy\\
		\email {giuseppe.liotta@unipg.it}
}

\authorrunning{W. Lenhart and G. Liotta}

	\maketitle

\begin{abstract}

 Let $\Gamma$ be a straight-line drawing of a graph and let $u$ and $v$ be two vertices of $\Gamma$. The
Gabriel disk of $u,v$ is the disk having $u$ and $v$ as antipodal points. A pair $\langle \Gamma_0,\Gamma_1 \rangle$ of
vertex-disjoint straight-line drawings form a mutual witness Gabriel drawing when, for $i=0,1$, any two vertices $u$ and $v$ of $\Gamma_i$ are adjacent if and only if their Gabriel disk does not contain any vertex of $\Gamma_{1-i}$. We characterize the pairs $\langle G_0,G_1 \rangle $ of complete bipartite graphs that admit a mutual witness Gabriel drawing. The characterization leads to a linear time testing algorithm. We also show that when at least one of the graphs in the pair $\langle G_0, G_1 \rangle $ is complete $k$-partite with $k>2$ and all partition sets in the two graphs have size greater than one, the pair does not admit a mutual witness Gabriel drawing.

\end{abstract}

\section{Introduction}\label{se:intro}

Proximity drawings, including Delaunay triangulations, rectangle of influence drawings, minimum spanning trees, and unit disk graphs, are among the most studied geometric graphs. They are commonly used as descriptors of the ``shape'' of a point set and are used in a variety of applications, including machine learning, pattern recognition, and computer graphics (see, e.g.,~\cite{DBLP:books/wi/CGHandbook2017}).
They have also been used to measure the faithfulness of large graph visualizations (see, e.g.,~\cite{DBLP:journals/jgaa/EadesH0K17}).

Proximity drawings are geometric graphs in which two vertices are adjacent if and only if they are deemed close by some measure. A common approach to define the closeness of two vertices $u$ and $v$ uses a \emph{region of influence} of $u$ and $v$, which is a convex region whose shape depends only on the relative position of $u$ with respect to $v$. Then we say that $u$ and $v$ are adjacent if and only if their region of influence does not contain some obstacle, often another vertex of the drawing. For example, a \emph{Gabriel drawing} $\Gamma$ is a proximity drawing where the region of influence  $u$ and $v$ is the disk having $u$ and $v$ as antipodal points, called the Gabriel region of $u$ and $v$; $u$ and $v$ are adjacent in $\Gamma$ if and only if their Gabriel region does not contain any other vertex. See also~\cite{DBLP:reference/crc/Liotta13} for a survey on different types of proximity regions and drawings.

An interesting generalization of proximity drawings is given in a sequence of papers  by Aronov, Dulieu, and Hurtado who introduce and study \emph{witness proximity drawings} and \emph{mutual witness proximity drawings}~\cite{DBLP:journals/comgeo/AronovDH11,DBLP:journals/comgeo/AronovDH13,DBLP:journals/ipl/AronovDH14,DBLP:journals/gc/AronovDH14}. In a witness proximity drawing the obstacles are points, called witnesses, that are suitably placed in the plane to impede the existence of  edges between non-adjacent vertices; these points may or may not include some of the vertices of the drawing itself. A mutual witness proximity drawing is a pair of  witness proximity drawings that are computed simultaneously and such that the vertices of one drawing are the witnesses of the other drawing.
For example, Figure~\ref{fi:introfigure} depicts a mutual witness Gabriel drawing (MWG-drawing for short) of two trees. In the figure, the Gabriel disk of $v_0,v_1$ of $\Gamma_0$ includes vertex $v_2$ but no vertices of $\Gamma_1$ and hence $v_0,v_1$ are adjacent in $\Gamma_0$; conversely, $v_1$ and $v_2$ are not adjacent in $\Gamma_0$ because their Gabriel disk contains vertex $u_1$ of $\Gamma_1$.

\begin{figure}[t]
	\centering
		\includegraphics[width=0.5\textwidth, page=1]{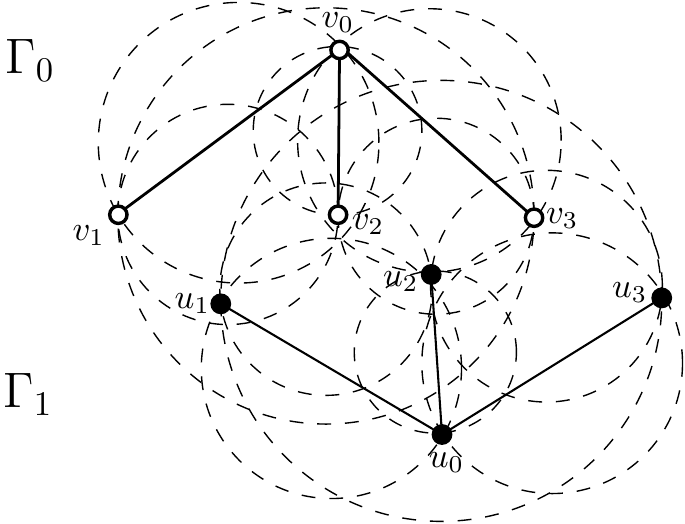}
			\caption{A mutual witness Gabriel drawing of two trees (Gabriel disks are dotted).}
	\label{fi:introfigure}
\end{figure}

In this  paper we characterize those pairs of complete bipartite graphs that admit an  MWG-drawing. While every complete bipartite graph has a witness Gabriel drawing~\cite{DBLP:journals/comgeo/AronovDH13}, not all pairs of complete bipartite graphs admit an MWG-drawing. To characterize the drawable pairs we also investigate some properties of MWG-drawings that go beyond complete bipartiteness. More precisely:

\begin{itemize}
\item We show that if $\langle \Gamma_0,\Gamma_1 \rangle $ is an MWG-drawing such that both $\Gamma_0$ and $\Gamma_1$ have diameter two, then the set of vertices of $\Gamma_0$ is linearly separable from the set of vertices of $\Gamma_1$. This extends a result of~\cite{DBLP:journals/ipl/AronovDH14}, where linear separability is proved when the diameter is one, i.e. when the two graphs are complete.
\item We show, perhaps surprisingly, that if $\langle G_0,G_1 \rangle $  is a pair of complete bipartite graphs that admits an MWG-drawing, then  both must be planar.
\item The above result let us characterize those pairs $\langle G_0, G_1 \rangle $ of complete bipartite graphs that admit an MWG-drawing 
    and leads to a linear time testing algorithm. When the test returns that $\langle G_0, G_1 \rangle $ is drawable, an MWG-drawing can be constructed in linear time in the real RAM model.
\item We show that relaxing the bipartiteness assumption does not significantly enlarge the class of representable graph pairs: We consider those pairs of complete multi-partite graphs each having all partition sets of size at least two and prove that if at least one of the graphs in the pair has more than two partition sets, then the pair does not admit an MWG-drawing.

\end{itemize}

We remark that our contribution not only fits into the rich literature devoted to proximity drawings, but it also relates to two other well studied topics in graph drawing, namely simultaneous embeddings (see, e.g.,~\cite{DBLP:reference/crc/BlasiusKR13,DBLP:books/sp/20/Rutter20} for references) and obstacle representations (see, e.g.,~\cite{DBLP:journals/dcg/AlpertKL10,DBLP:journals/dcg/BalkoCV18,DBLP:conf/gd/ChaplickLPW16,DBLP:journals/combinatorics/DujmovicM15,https://doi.org/10.48550/arxiv.2202.13015,DBLP:conf/cccg/JohnsonS14,DBLP:journals/combinatorics/MukkamalaPP12,DBLP:conf/wg/MukkamalaPS10,DBLP:journals/gc/PachS11}). As in simultaneous embeddings, the coordinates of the vertices of $\Gamma_i$ in a mutual witness proximity drawing are defined by taking into account the (geometric and topological) properties of $\Gamma_{1-i}$; as in obstacle graph representations, the adjacency of the vertices $\Gamma_i$ depends on whether their geometric interaction is obstructed by some external obstacles, namely the vertices of $\Gamma_{1-i}$;. Finally, mutual witness proximity drawings are of interest in pattern recognition, where they have been used in the design of trained classifiers to convey information about the interclass structure of two sets of features (see, e.g.~\cite{DBLP:journals/pr/IchinoS85}).

\section{Preliminaries}\label{se:preli}
We assume familiarity with basic definitions and results of graph drawing~\cite{DBLP:books/ph/BattistaETT99}.
%
We assume that all drawings occur in the Euclidean plane with standard $x$ and $y$ axes, and so concepts such as above/below a (non-vertical) line are unambiguous.
Given two distinct points $p$ and $q$ in the plane, we denote by $\overline{pq}$ the straight-line segment whose endpoints are $p$ and $q$. Also, let $a,b,c$ be three distinct points in the Euclidean plane, we denote by $\Delta(abc)$ the triangle whose vertices are $a,b,c$. Given two non-axis-parallel lines $\ell_1$ and $\ell_2$ intersecting at a point $b$, those lines divide the plane into four \emph{wedges: the top, bottom, left, and right wedges of $b$ with respect to $\ell_1$ and $\ell_2$}. The top and bottom wedges lie entirely above and below the horizontal line through $b$, respectively; the left and right wedges lie entirely to the left and right of the vertical line through $b$. When the two lines are determined by providing a point (other than $b$) on each line, say $a$ and $c$, we denote the wedges by $W_T[b, a, c]$, $W_B[b, a, c]$, $W_L[b, a, c]$, and $W_R[b, a, c]$ when we want to include the boundary of each wedge as part of that wedge and by $W_T(b, a, c)$, $W_B(b, a, c)$, $W_L(b, a, c)$, and $W_R(b, a, c)$ when we do not.

Note that exactly one of the four wedges will have both $a$ and $c$ on its boundary, we denote that wedge as $W[b, a, c]$ (or $W(b, a, c)$).
See Figure~\ref{fi:wedge}.

\begin{figure}[t]
		\centering
		\subfigure[]{\label{fi:wedge}\includegraphics[width=0.4\columnwidth]{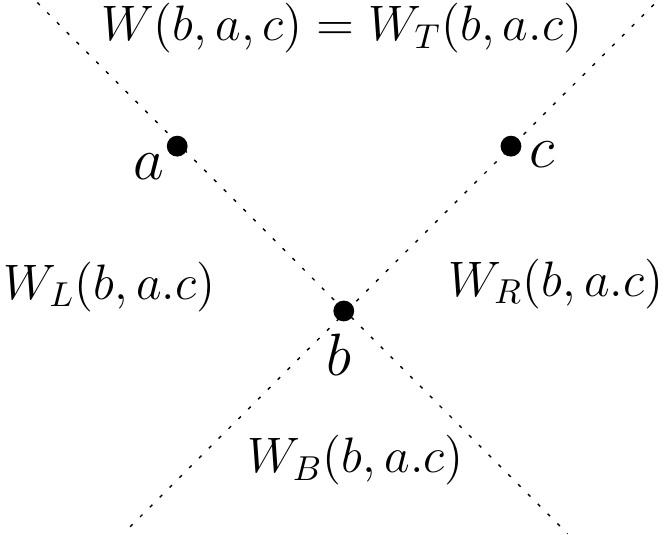}}
        \hfill
		\subfigure[]{\label{fi:triangle}\includegraphics[width=0.45\columnwidth]{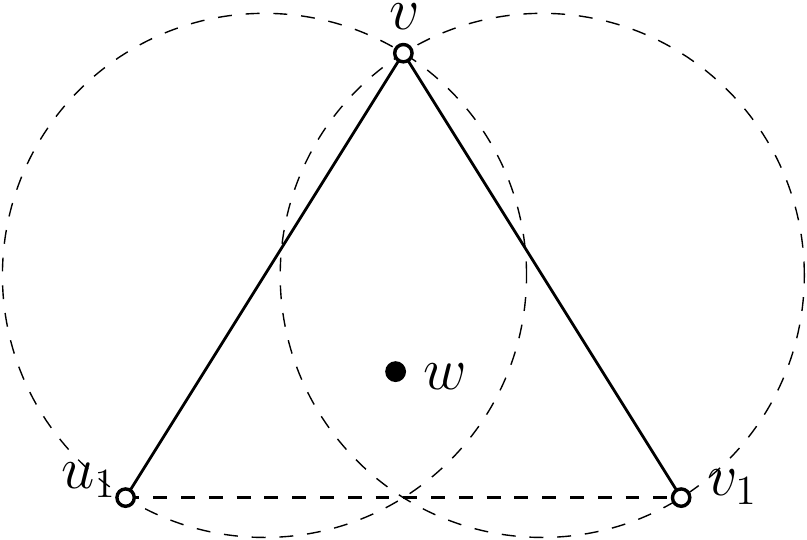}}
		\caption{(a) $W(b,a,c) = W_T(b,a,c)$, $W_B(b,a,c)$, $W_L(b,a,c)$, and $W_R(b,a,c)$ (b) If $w \in \Delta(vu_1v_1)$, at most one of $\overline{vu_1}$ and $\overline{vv_1}$ is an edge of a WG-drawing.}\label{fi:preli}
	\end{figure}


Let $\Gamma$ be a straight line drawing of a graph $G$ and let $u$ and $v$ be two vertices of $\Gamma$ (and of $G$). Vertices $u$ and $v$ may either be adjacent in $G$ and thus $\overline{uv}$ is an edge of $\Gamma$ or $u$ and $v$ are not adjacent vertices, in which case $\overline{uv}$ is a \emph {non-edge} of $\Gamma$. For example, $\overline{v_0v_1}$ in \cref{fi:introfigure} is an edge while $\overline{v_1v_2}$ is a non-edge of $\Gamma_1$. Also, the \emph{Gabriel disk} of $p$ and $q$, denoted as $D[p, q]$ is the disk having $p$ and $q$ as antipodal points; $D[p, q]$ is a closed set.

Let $V$ and $P$ be two sets of distinct points in the plane. A \emph {witness Gabriel drawing} ({\em WG-drawing}) with vertex set $V$ and witness set $P$ is a geometric graph $\Gamma$ whose vertices are the points of $V$ and such that any two vertices $u$ and $v$ form an edge if and only if $D[u, v] \cap P = \emptyset$. A graph $G$ is \emph {witness Gabriel drawable} ({\em WG-drawable}) if there exist two point sets $V$ and $P$ such that the witness Gabriel drawing with vertex set $V$ and witness set $P$ represents $G$ (i.e., there is a bijection between the vertex set of $G$ and the point set $V$ and between the edge set of $G$ and the edge set of $\Gamma$ that is incidence-preserving). The following property can be proved with elementary geometric arguments (see also Figure~\ref{fi:triangle}).


\begin{property}\label{pr:triangle}
Let $\Gamma$ be a WG-drawing with witness set $P$ and let $\overline{vu_1}$ and $\overline{vv_1}$ be two edges of  $\Gamma$ incident on the same vertex $v$. Then $\Delta(vu_1v_1) \cap P = \emptyset$.
\end{property}

For a pair $\langle G_0,G_1 \rangle $ of WG-drawable graphs, a \emph{mutual witness Gabriel drawing} ({\em MWG-drawing}) is a pair $\langle \Gamma_0,\Gamma_1 \rangle $ of straight-line drawings such that $\Gamma_i$ is a WG-drawing of $G_i$ with witness set the vertices of $\Gamma_{1-i}$ ($i=0,1$). If $\langle G_0,G_1 \rangle $ admits an MWG-drawing we say that $\langle G_0,G_1 \rangle $ is \emph{mutually witness Gabriel drawable} (\emph{MWG-drawable}).

Details of proofs of statements marked with '*' can be found in the appendix.

%
%
%
%


\section{Linear Separability of Diameter-2 MWG-drawings}\label{se:diameter-2}

In this section we extend a result by Aronov et al.~\cite{DBLP:journals/ipl/AronovDH14} about the linear separability of the MWG-drawings of complete graphs to graphs of diameter two.


\begin{lemma}\label{le:non-crossing}
Let $\langle \Gamma_0, \Gamma_1 \rangle $ be an MWG-drawing such that $\Gamma_i$ has diameter at most $2$ ($i=0,1$).
Then no segment of $\Gamma_0$ intersects any segment of $\Gamma_1$.
\end{lemma}

\begin{proof}
Note first that by Property~\ref{pr:triangle}, a vertex $u$ of $\Gamma_i$ cannot lie on a non-edge $\overline{u_1, v_1}$ of $\Gamma_{1-i}$
since $\{u_1, v_1\}$ have at least one common neighbor $v$. Also no vertex of $\Gamma_i$ can lie on an edge of $\Gamma_{1-i}$.
Let $\overline{u_0v_0}$ be an edge of $\Gamma_0$ and let $\overline{u_1v_1}$ be an edge of $\Gamma_1$. Assume that they cross and consider the quadrilateral $Q$ whose vertices are the end-points of the two crossing edges. Since some internal angle of $Q$ must be at least $\frac{\pi}{2}$, either $D[u_0, v_0]$ contains one of $\{u_1,v_1\}$ or $D[u_1, v_1]$ contains one of $\{u_0,v_0\}$ contradicting the fact that both $\overline{u_0v_0}$ is an edge of $\Gamma_0$ and $\overline{u_1v_1}$ is an edge of $\Gamma_1$.

Let $\overline{u_0v_0}$ be an edge of $\Gamma_0$ and let $\overline{u_1v_1}$ be a non-edge of $\Gamma_1$. Since $\Gamma_1$ has diameter at most two, there is a vertex $v$ in $\Gamma_1$ such that both $\overline{vu_1}$ and $\overline{vv_1}$ are edges of $\Gamma_1$. Since $\overline{u_1 v_1}$ crosses $\overline{u_0 v_0}$, but neither $\overline{vu_1}$ nor $\overline{vv_1}$ crosses $\overline{u_0v_0}$, we have that one of $\{u_0,v_0\}$ is a point of $\Delta(u_1,v_1,v)$. However, $\Gamma_1$ is a WG-drawing whose witness set is the set of vertices of $\Gamma_0$ and, by Property~\ref{pr:triangle}, no vertex of $\Gamma_0$ can be a point of $\Delta(u_1,v_1,v)$. An analogous argument applies when $\overline{u_0v_0}$ is a non-edge of $\Gamma_0$ while $\overline{u_1v_1}$  is  an edge of $\Gamma_1$.


It remains to consider the case that $\overline{u_0v_0}$ is a non-edge of $\Gamma_0$ and $\overline{u_1v_1}$ is a non-edge of $\Gamma_1$. Let $v$ be a vertex such that both $\overline{vu_1}$ and $\overline{vv_1}$ are edges of $\Gamma_1$. By the previous case, neither of these two edges can cross $\overline{u_0v_0}$. It follows that one of $\{u_0,v_0\}$ is a point of $\Delta(u_1,v_1,v)$ which, by Property~\ref{pr:triangle}, is impossible.
\end{proof}

\begin{figure}[t]
	\centering
		\includegraphics[width=0.8\textwidth, page=1]{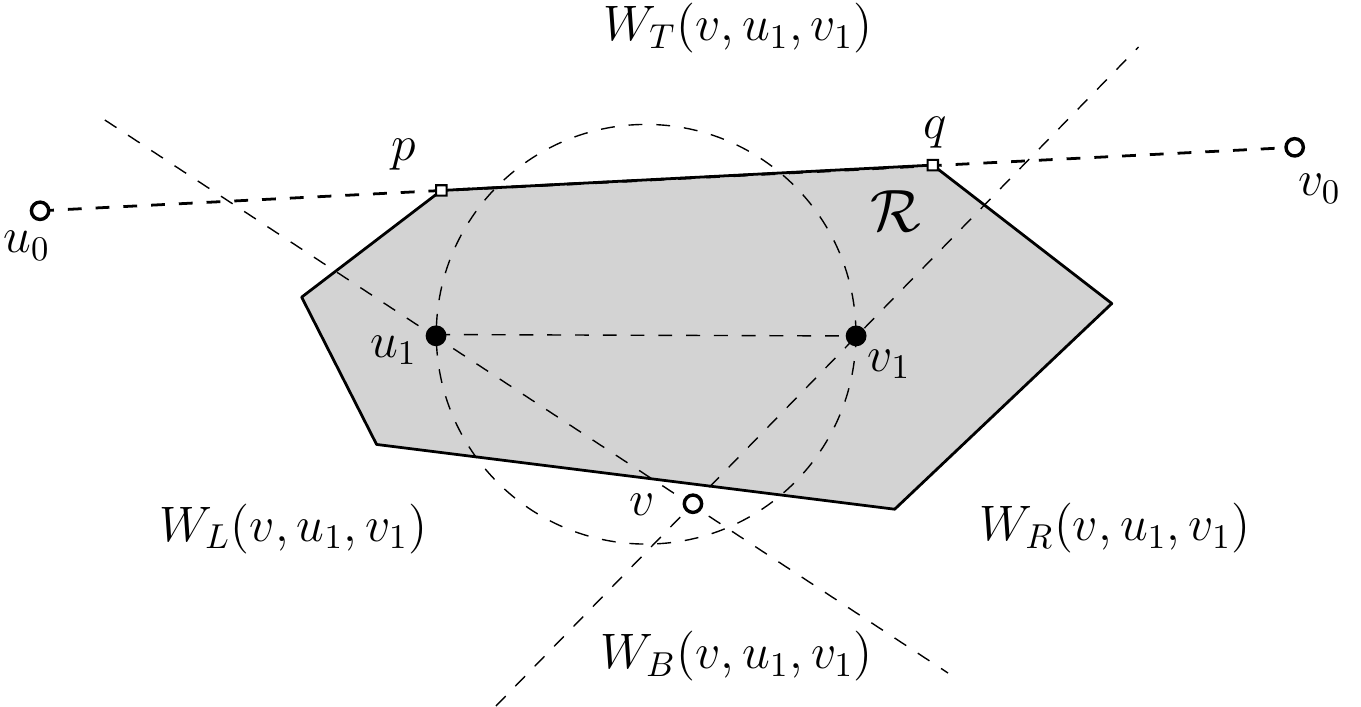}
			\caption{Illustration for the proof of \cref{th:linear-separability}: a non-edge $\overline{u_1v_1}$ of $\Gamma_1$ in a region $\mathcal{R}$ bounded by portions of edges and non-edges of $\Gamma_0$.}
	\label{fi:linear-separability}
\end{figure}

We are now ready to prove the main result of this section. We denote by $\mathit{CH}(\Gamma)$ the convex hull of the vertex set of a drawing $\Gamma$.

\begin{theorem}\label{th:linear-separability}
Let $\langle \Gamma_0, \Gamma_1 \rangle $ be an MWG-drawing such that each $\Gamma_i$ has diameter $2$ . Then $\Gamma_0$ and $\Gamma_1$ are linearly separable.
\end{theorem}
\begin{proof}
By \cref{le:non-crossing}, no vertex, edge or non-edge of $\Gamma_i$ intersects any vertex, edge or non-edge of $\Gamma_{1-i}$.
Hence, either $\mathit{CH}(\Gamma_0)$ and $\mathit{CH}(\Gamma_1)$ are linearly separable and we are done, or one of the convex hulls -- say $\mathit{CH}(\Gamma_1)$ -- is contained in a convex region $\mathcal{R}$ bounded by (portions of) edges and/or non-edges of $\Gamma_0$. We prove that region $\mathcal{R}$ cannot exist, which implies the statement.
Suppose for a contradiction that $\Gamma_1$ is contained in $\mathcal{R}$ and let $\overline{u_1v_1}$ be a non-edge of $\Gamma_1$ with $x(u_1) \leq x(v_1)$. See Figure~\ref{fi:linear-separability} for a schematic illustration.  Since $u_1$ and $v_1$ are not adjacent, there is some vertex $v$ of $\Gamma_0$ such that $v \in D[u_1, v_1]$.
Without loss of generality, assume that $\overline{u_1v_1}$ is horizontal and that $v$ is below the line through $u_1$ and $v_1$. Since $u_1$ and $v_1$ are points of $W_T[v, u_1, v_1]$ and $\Gamma_1$ is contained in $\mathcal{R}$, there is some segment $\overline{pq}$ of the boundary of $\mathcal{R}$  such that $\overline{pq}$ intersects $W_T(v, u_1, v_1)$ above the line through $u_1$ and $v_1$.  Let $u_0$ and $v_0$ be the two vertices of $\Gamma_0$ such that $\overline{pq}$ is a subset of $\overline{u_0v_0}$.
For concreteness, we assume that the $x$-coordinates of $u_0$, $p$, $q$, and $v_0$ are such that $x(u_0) \leq x(p) \leq x(q) \leq x(v_0)$.

\begin{claim}\label{cl:opposite-wedges}
$u_0 \in W_L(v, u_1, v_1)$ , $v_0 \in W_R(v, u_1, v_1)$ and $\overline{u_0v_0}$ is a non-edge of $\Gamma_0$.
\end{claim}
\noindent \emph{Proof of the claim:} Suppose for a contradiction that a vertex of $\{u_0,v_0\}$ -- say $v_0$ -- were a point of $W_T[v,u_1,v_1]$. Since $\overline{pq}$ intersects $W_T(v, u_1, v_1)$ above the horizontal line through $u_1$ and $v_1$, we have that  $v_0$ must also be above this horizontal line or else $\overline{u_0v_0}$ and $\overline{u_1v_1}$ would cross, contradicting \cref{le:non-crossing}. However, if $v_0$ is above the line through $u_1$ and $v_1$ we have that $\overline{u_1v_1}$  and $\overline{vv_0}$ cross which again contradicts \cref{le:non-crossing}. Therefore, $v_0 \not \in W_T[v, u_1, v_1]$  and, by the same argument, $u_0 \not \in W_T[v, u_1, v_1]$. Note that this argument also precludes either point of $\{u_0,v_0\}$ from
being in $W_B[v, u_1, v_1]$, since, because $\overline{pq}$ intersects $W_T[v, u_1, v_1]$, we would then have that the other point of
$\{u_0,v_0\}$ lies in $W_T[v, u_1, v_1]$.
Finally, observe that if $u_0$ and $v_0$ were both points of either $W_L(v, u_1, v_1)$ or $W_R(v, u_1, v_1)$, segment $\overline{pq}$ would not intersect $W_T(v, u_1, v_1)$. It follows that $u_0 \in W_L(v, u_1, v_1)$  and $v_0 \in W_R(v, u_1, v_1)$.
Note that $\triangle(u_0vv_0)$ contains both $u_1$ and $v_1$, so $\angle u_0 v v_0 > \angle u_1 v v_1 \geq \frac{\pi}{2}$, which implies that $\angle u_0 v_1 v_0 > \frac{\pi}{2}$ and so $\overline{u_0v_0}$ is a non-edge of $\Gamma_0$. This concludes the proof of the claim.

\begin{figure}[t]
		\centering
		\subfigure[]{\label{fi:diameter-not-two-a}\includegraphics[width=0.4\columnwidth]{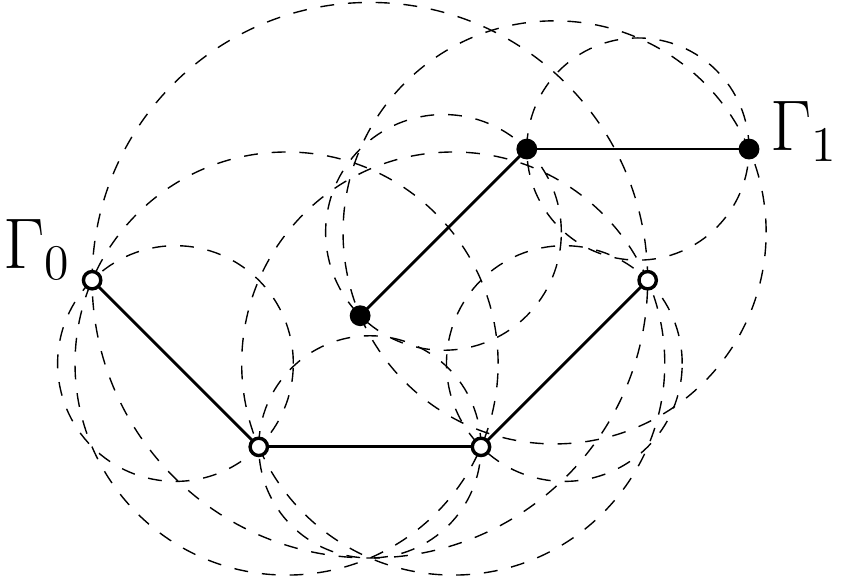}}
        \hfill
		\subfigure[]{\label{fi:diameter-not-two-b}\includegraphics[width=0.4\columnwidth]{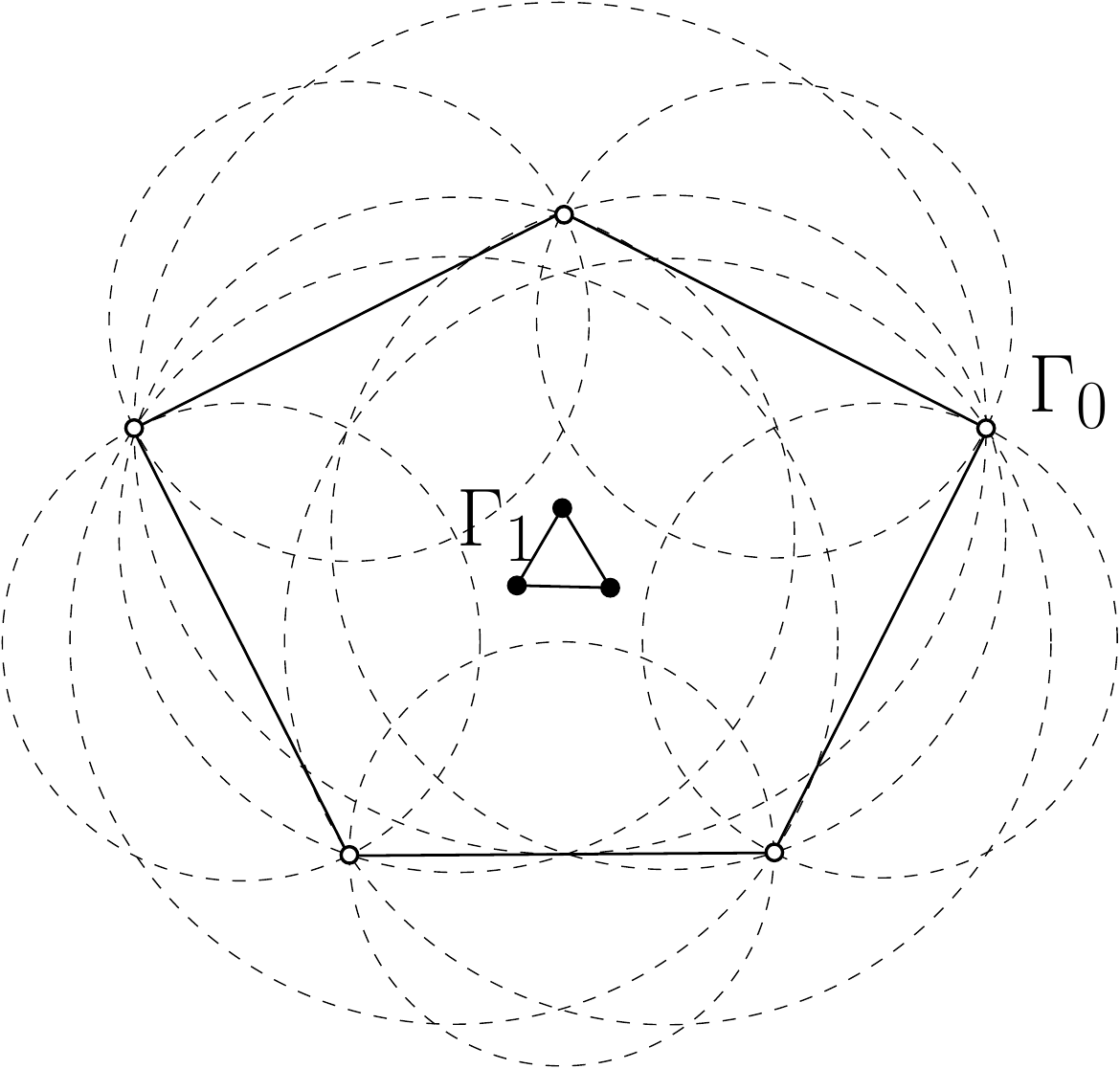}}
		\caption{Non-linearly separable MWG-drawings of: (a) a diameter two graph and a diameter three graph; (b) a diameter two graph and a diameter one graph.}\label{fi:dimater-not-two}
	\end{figure}

\medskip

By the claim above and by the assumption that $\Gamma_0$ has diameter two, there is some vertex $z$ such that both $\overline{zu_0}$ and $\overline{zv_0}$ are edges of $\Gamma_0$. Vertex $z$ may or may not coincide with $v$. If $z$ coincides with $v$ or if $z \in W_B[v, u_1, v_1]$, we have that $\Delta(z u_0 v_0)$ contains both $u_1$ and $v_1$ and two of its sides are edges of $\Gamma_0$, which contradicts Property~\ref{pr:triangle}.
If $z \in W_T[v, u_1, v_1]$ and $z$ is above the line through $u_1$ and $v_1$, we have that $\overline{vz}$ and $\overline{u_1v_1}$ cross, which contradicts \cref{le:non-crossing}. If $z \in W_T[v, u_1, v_1]$ and $z$ is below the line through $u_1$ and $v_1$, either  $\overline{u_1v_1}$ crosses one of $\{\overline{zu_0}, \overline{zv_0}\}$ contradicting \cref{le:non-crossing}, or $\Delta(zu_0v_0)$ contains both $u_1$ and $v_1$ contradicting Property~\ref{pr:triangle}.
If $z \in W_L(v, u_1, v_1)$, we consider three cases. If edge $\overline{zv_0}$ crosses $\overline{u_1v_1}$, we would violate \cref{le:non-crossing}. If edge $\overline{zv_0}$ is above $\overline{u_1v_1}$, then $\angle{zvv_0} > \frac{\pi}{2}$ and since both $u_1$ and $v_1$ are in the interior of $\triangle(zvv_0)$, we have that $u_1$ and $v_1$ are in $D[z,v_0]$, contradicting the fact that $\overline{zv_0}$ is an edge of $\Gamma_0$.   If edge $\overline{zv_0}$ is below $\overline{u_1v_1}$,  $\triangle(u_0zv_0)$ contains both $u_1$ and $v_1$, which violates Property~\ref{pr:triangle}. By a symmetric argument, we have that $z$ cannot be a point of $W_R(v, u_1, v_1)$ either. Since point $z$ does not exist, it follows that $\mathcal{R}$ does not exist.
\end{proof}

\cref{th:linear-separability} shows that  MWG-drawings with diameter two capture useful information about the interaction of two point sets. As pointed out by both Ichino and Slansky~\cite{DBLP:journals/pr/IchinoS85} and  by Aronov et al.~\cite{DBLP:journals/ipl/AronovDH14}, the linear separability of mutual witness proximity drawings gives useful information about the interclass structure of two set of points. It is also worth noting that if at least one of the graphs in the pair has diameter different from two, a non-linearly separable drawing may exist. For example Figure~\ref{fi:diameter-not-two-a} and Figure~\ref{fi:diameter-not-two-b} show MWG-drawings of graph pairs in which the diameter two property is violated for one of the two graphs.


\section{MWG-drawable Complete Bipartite Graphs}\label{se:bipartite}

In this section we exploit \cref{th:linear-separability} to characterize those pairs of complete bipartite graphs that admit an MWG-drawing. In Section~\ref{sse:planarity} we prove that any two complete bipartite graphs that form an MWG-drawable pair are planar. The complete characterization is then given in Section~\ref{sse:characterization}. In what follows we shall assume without loss of generality that the line separating a drawing $\Gamma$ from its set of witnesses is horizontal and it coincides with the line $y=0$, with the witnesses in the negative half-plane.
The proof of the following property is trivial and therefore omitted, but Figure~\ref{fi:in-the-wedge} and its caption illustrate it.

\begin{figure}[ht]
		\centering
		\subfigure[]{\label{fi:in-the-wedge}\includegraphics[width=0.4\columnwidth]{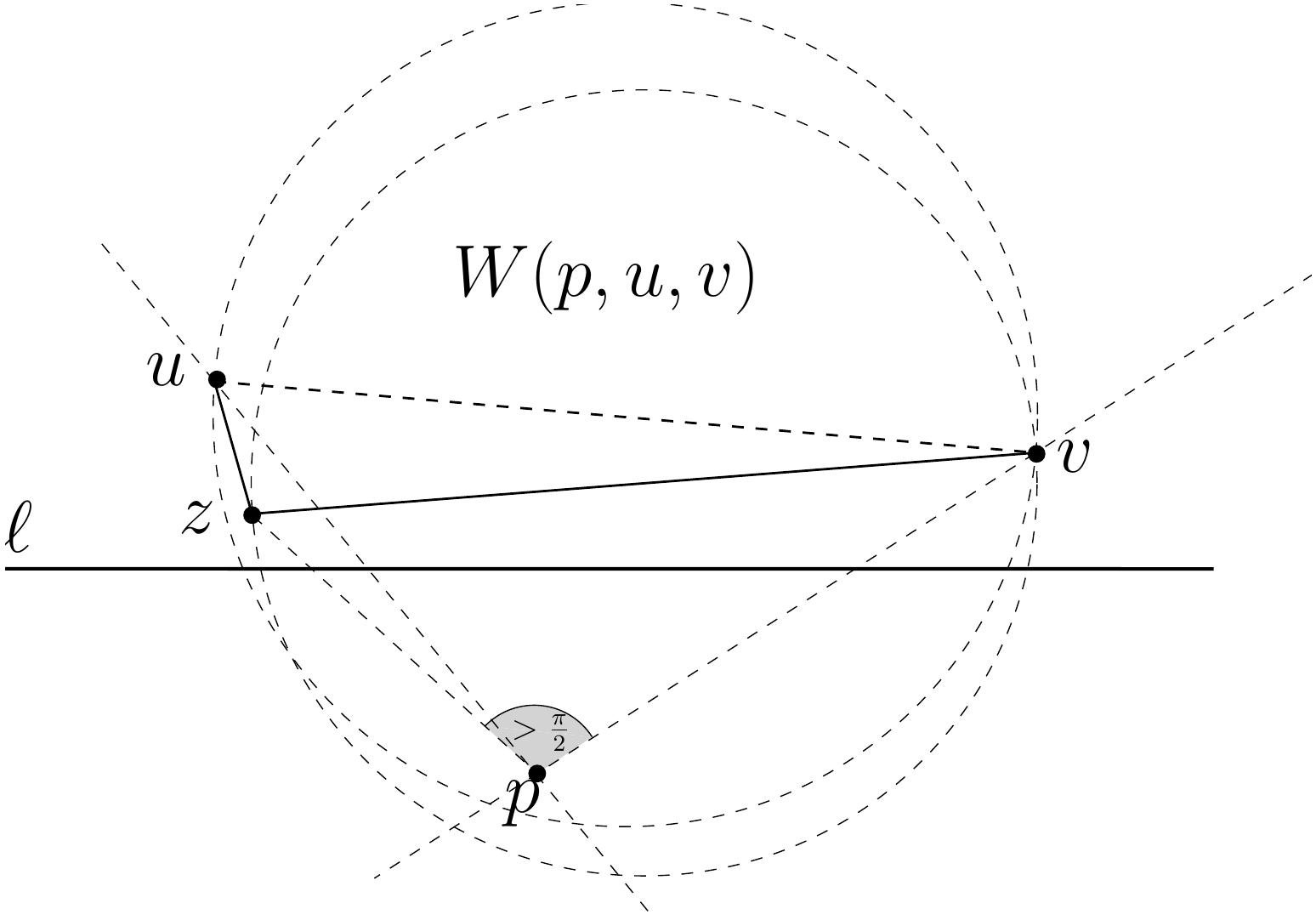}}
\hfill
		\subfigure[]{\label{fi:alternating}\includegraphics[width=0.5\columnwidth]{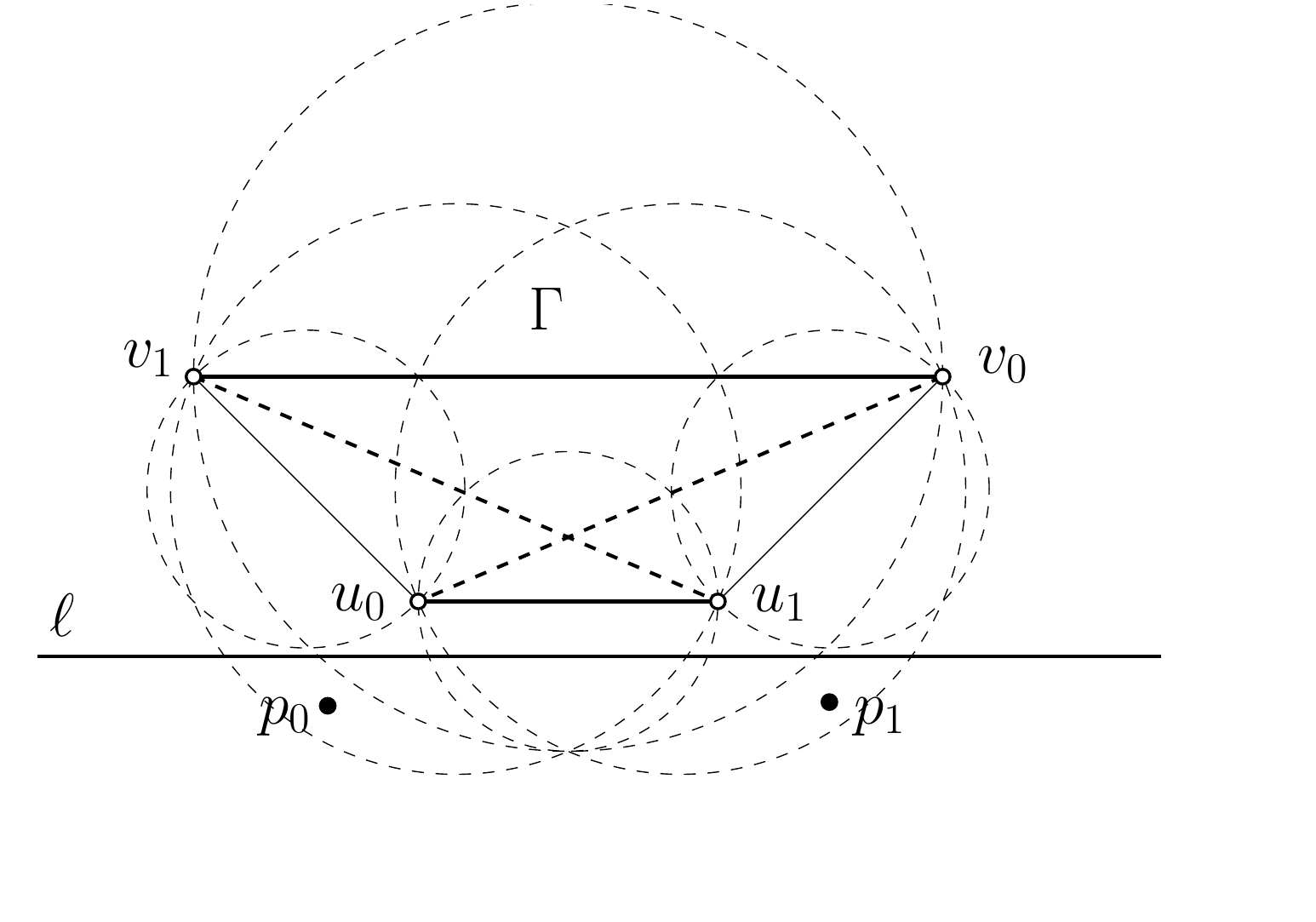}}
		\caption{(a) If $z \not \in W(p, u, v)$, then $p \in D[z,v]$ and $\overline{zv}$ is not an edge of $\Gamma$. (b) A WG-drawing $\Gamma$ with an alternating 4-cycle highlighted in bold.}\label{fi:universal}
	\end{figure}

\begin{property}\label{pr:common-neighbor}
Let $\Gamma$ be a WG-drawing with witness set $P$, let $\overline{uv}$ be a non-edge of $\Gamma$ with witness $p \in P$, and let $z$ be a vertex of $\Gamma$ such that both $\overline{zu}$ and $\overline{zv}$ are edges of $\Gamma$. Then $z \in W(p, u, v)$.
\end{property}

\subsection{Planarity}\label{sse:planarity}

Let $\Gamma$ be a WG-drawing; an \emph{alternating 4-cycle} in  $\Gamma$ consists of two vertex-disjoint edges $\overline{u_0u_1}$ and $\overline{v_0v_1}$ of $\Gamma$ such that $\overline{u_0v_0}$ and $\overline{u_1v_1}$ are both non-edges in the drawing. For example, Figure~\ref{fi:alternating} shows a WG-drawing $\Gamma$ whose witness set consists of points $p_0$ and $p_1$. In the figure, $\overline{u_0u_1}$ and $\overline{v_0v_1}$ are edges of $\Gamma$ while $\overline{u_0v_0}$ and $\overline{u_1v_1}$ are non-edges of $\Gamma$: these two pairs of edges  and non-edges (bolder in the figure) form an alternating 4-cycle in $\Gamma$.


\begin{lemma}\label{le:alternating-4-cycle}
Let $\Gamma$ be a WG-drawing of a complete bipartite graph such that $\Gamma$ is linearly separable from its witness set $P$ and let $C$ be an alternating 4-cycle defined on $\Gamma$. The two edges of $\Gamma$ in $C$ do not cross while the two non-edges of $\Gamma$ in $C$ do cross.
\end{lemma}

\begin{proof} Let  $\ell$ be the line separating $\Gamma$ from its witness set. Let $u_0,u_1,v_1,v_0$ be the four vertices of $C$ such that $\overline{u_0u_1}$ and $\overline{v_0v_1}$ are two edges of $\Gamma$ while $\overline{u_0v_0}$ and $\overline{u_1v_1}$ are two non-edges of $\Gamma$.  Since the drawing is a complete bipartite graph,  $\overline{u_0v_1}$ and $\overline{v_0u_1}$ are edges of $\Gamma$. We prove that $\overline{u_0v_0}$ and $\overline{u_1v_1}$ must cross in $\Gamma$, which implies that $\overline{u_0u_1}$ and $\overline{v_0v_1}$ do not cross.

Let $p_0 \in P$ such that $p_0 \in D[u_0,v_0]$. By Property~\ref{pr:common-neighbor}, both $u_1$ and $v_1$ lie in the wedge $W(p_0, u_0, v_0) = W_T(p_0, u_0, v_0)$. Observe that $p_0$ cannot also be a witness for the pair $u_1$ and $v_1$ as otherwise, by Property~\ref{pr:common-neighbor}, we should have that also $u_0$ and $v_0$ lie in the top wedge $W_T(p_0, u_1, v_1)$, which is impossible. So, let $p_1 \in P$ be distinct from $p_0$ and such that $p_1 \in D[u_1,v_1]$. If $p_0$ were a point in $W_T[p_1, u_1, v_1]$, $p_0$ would also be a point in $\triangle(v_1p_1u_1)$ and we would have $p_0 \in D[u_1,v_1]$, which we just argued is impossible (see, e.g. Figure~\ref{fi:no-crossing}). By analogous reasoning we have that $p_1$ cannot be a point of $W_T[p_0, u_0, v_0]$.
Also, $p_1 \not \in W_B[p_0, u_0, v_0]$ or else $p_0$ would be in $W_T[p_1, u_1, v_1]$ since $W_T(p_1, u_1, v_1)$ contains both $u_0$ and $v_0$.
It follows that either $p_1 \in W_L(p_0, u_0, v_0)$ or $p_1 \in W_R(p_0, u_0, v_0)$. In either case, $W_T(p_1, u_1, v_1)$ can contain both $u_0$ and $v_0$ only if $\overline{u_0v_0}$ and $\overline{u_1v_1}$ cross.
 \end{proof}

\begin{figure}[ht]
	\centering
		\includegraphics[width=0.7\textwidth, page=1]{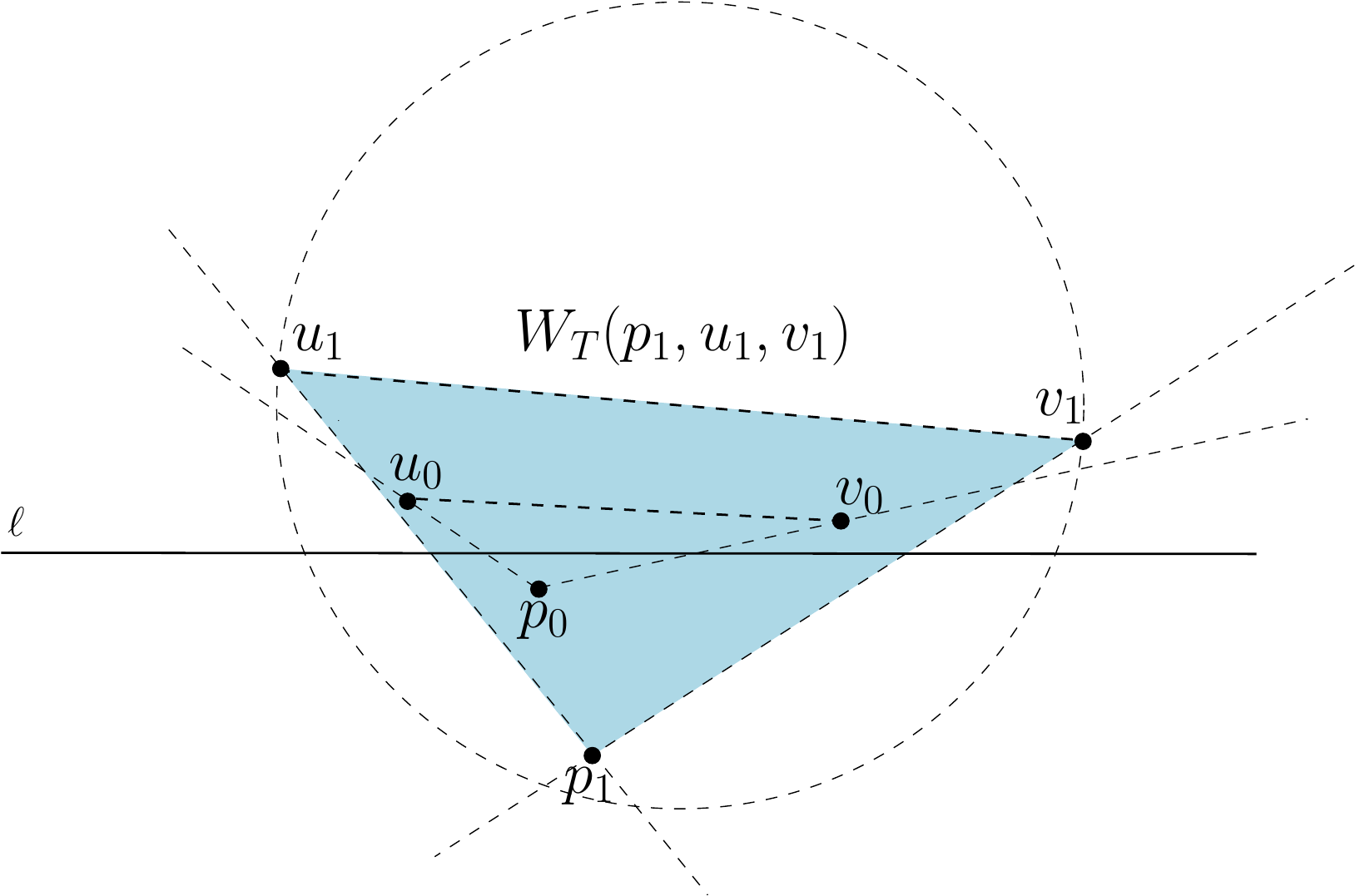}
			\caption{If $p_0 \in W_T[p_1, u_1, v_1]$, then $p_0 \in D[u_1,v_1]$.}
	\label{fi:no-crossing}
\end{figure}

%
%

The following corollaries are a consequence of \cref{le:alternating-4-cycle} and of \cref{th:linear-separability} .

\begin{restatable}[*]{corollary}{coplanarity}\label{co:planarity}
Let $G_0$ and $G_1$ be two vertex disjoint~complete~bipartite~graphs. If~the~pair~$\langle G_0,G_1 \rangle $~is~MWG-drawable,~then~both~$G_0$ and~$G_1$~are~planar~graphs.
\end{restatable}



\begin{restatable}[*]{corollary}{coconvexity}\label{co:4-cycle-convexity}
Let $\Gamma$ be a WG-drawing of a complete bipartite graph such that $\Gamma$ is linearly separable from its witness set. Any 4-cycle formed by edges of $\Gamma$ is a convex polygon.
\end{restatable}


\subsection{Characterization}\label{sse:characterization}

 We start with two technical lemmas.

\begin{figure}[t]
		\centering
		\subfigure[]{\label{fi:universal-a}\includegraphics[width=0.45\columnwidth]{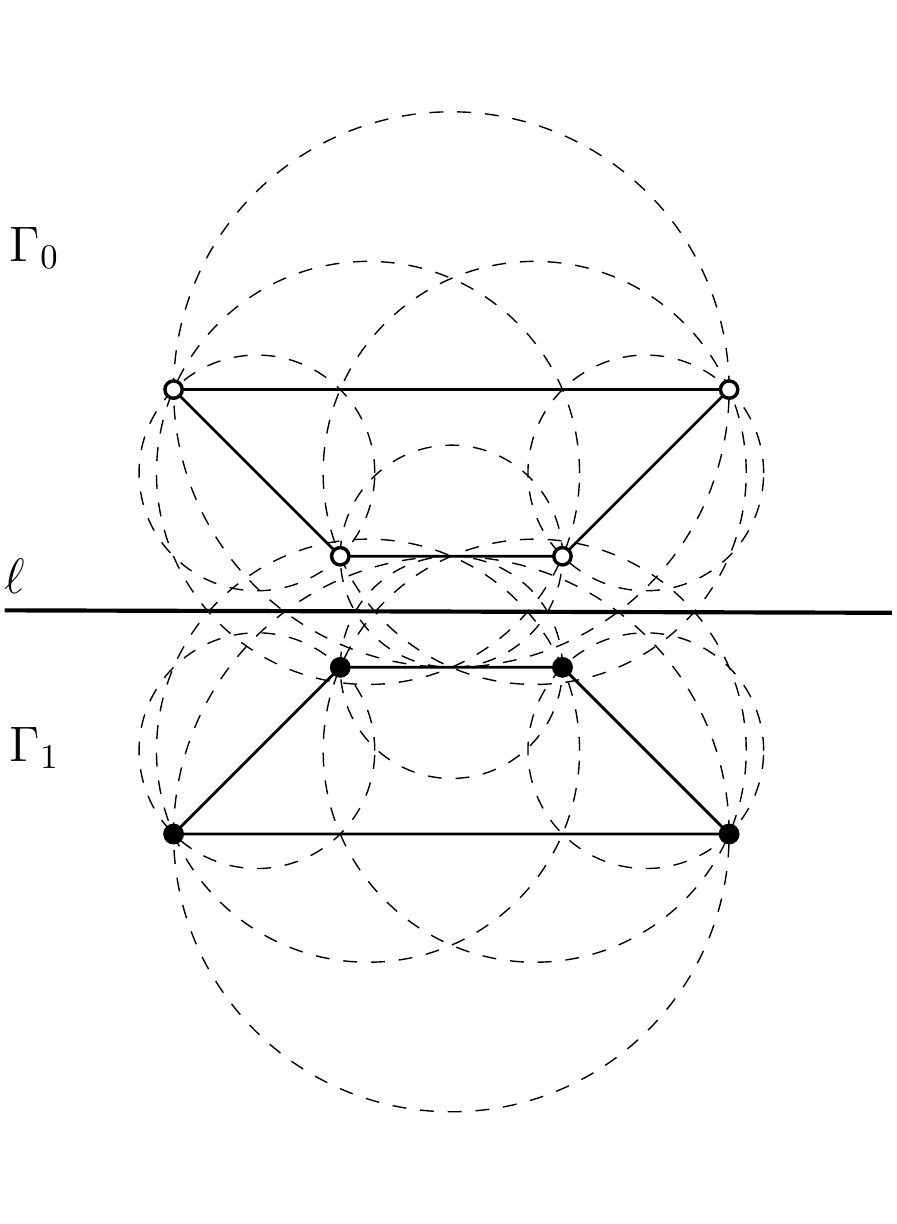}}
\hfill
		\subfigure[]{\label{fi:universal-b}\includegraphics[width=0.45\columnwidth]{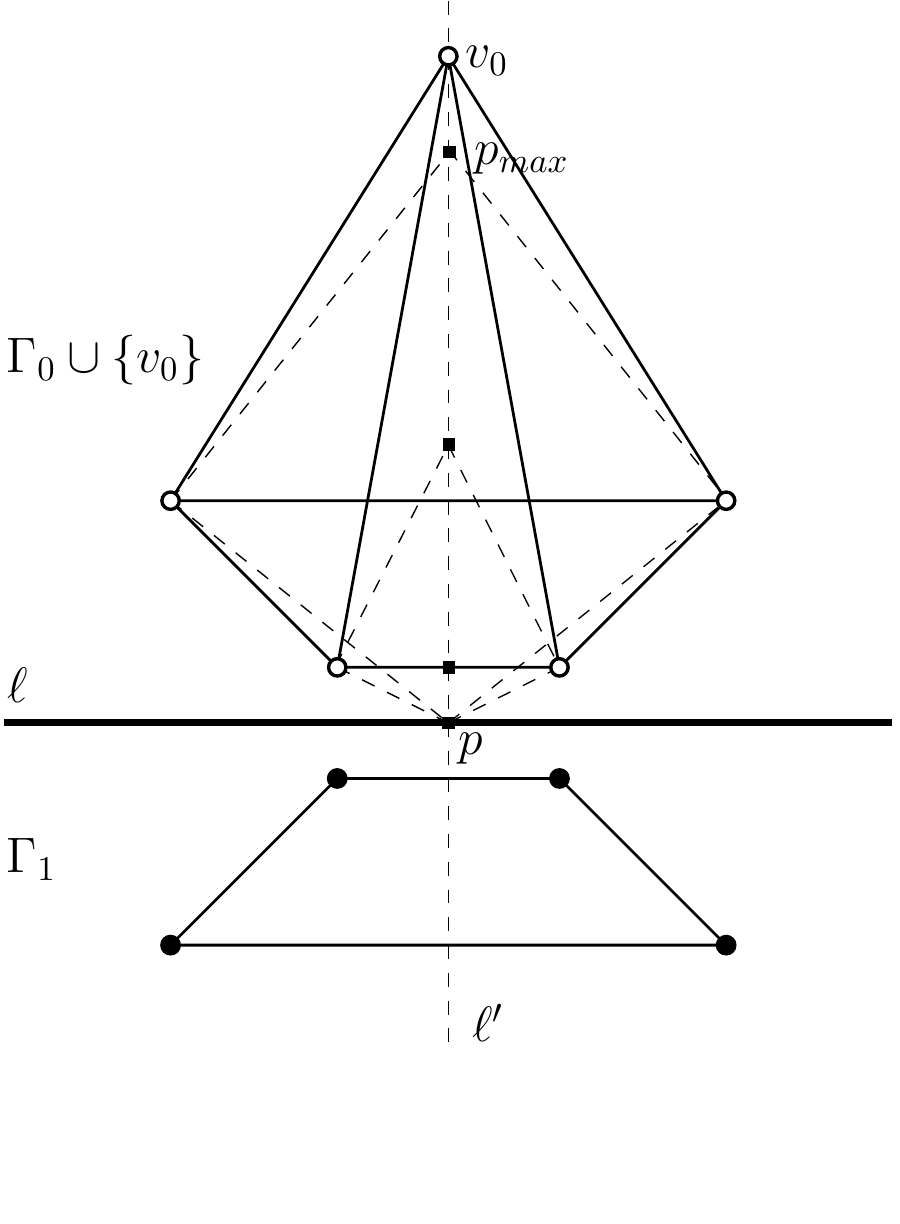}}
		\caption{(a) A linearly separable MWG-drawing $\langle \Gamma_0,\Gamma_1 \rangle $. (b)Adding a universal vertex $v_0$ to $\Gamma_0$ by placing it far enough from $p$ on $\ell'$.}\label{fi:universal}
	\end{figure}

\begin{restatable}[*]{lemma}{leuniversalvertex}\label{le:universal-vertex}
Let $\langle G_0,G_1 \rangle $ be an MWG-drawable pair admitting a linearly separable MWG-drawing. Then the pair $\langle G_0 \cup \{v_0\},G_1 \rangle $ also admits a linearly separable MWG-drawing, where $v_0$ is a vertex not in $G_0$ and is adjacent to all vertices in $G_0$---that is, a {\em universal vertex} of $G_0 \cup \{v_0\}$.
\end{restatable}

Figure~\ref{fi:universal-b} shows the addition of the universal vertex $v_0$ to the MWG-drawing of Figure~\ref{fi:universal-a}.
In fact, a universal vertex can be added to either drawing as long as it is positioned sufficiently far from the separating line.

%
%
%
%
%
Let $u_0,u_1$ be two points with $x(u_0) < x(u_1)$. The open \emph{vertical strip} of $u_0$, $u_1$, denoted as $S(u_0,u_1)$, is the set of points $(x,y)$ such that $x(u_0) < x < x(u_1)$. Assume now that $u_0$ and $u_1$ are vertices of a WG-drawing $\Gamma$ such that $\Gamma$ is linearly separable from its witness set by a line $\ell$. Segment $\overline{u_0u_1}$ divides $S(u_0,u_1)$ into two (open) {\em half-strips}: $S_N(u_0,u_1)$ is the (near) half-strip on the same side of $\overline{u_0u_1}$ as $\ell$ and $S_F(u_0,u_1)$ is the other (far) half-strip. $S[u_0,u_1]$, $S_N[u_0,u_1]$, and $S_F[u_0,u_1]$ consist of $S(u_0,u_1)$, $S_N(u_0,u_1)$, and $S_F(u_0,u_1)$ along with their respective boundaries.

\begin{restatable}[*]{lemma}{leisolatedvertex}\label{le:isolated-vertex}
	Let $\langle G_0,G_1 \rangle $ be an MWG-drawable pair admitting a linearly separable MWG-drawing. Then at least one of the pairs $\langle G_0 \cup \{v_0\},G_1 \rangle $, $\langle G_0 ,G_1\cup \{v_1\} \rangle $ also admits a linearly separable MWG-drawing, where, for $i = 0,1$, $v_i$ is a vertex not in $G_i$ and has no edges to any vertex in $G_i$---that is, $v_i$ is an isolated vertex of $G_i$.
\end{restatable}


\begin{figure}[t]
		\centering
		\subfigure[]{\label{fi:isolated-a}\includegraphics[width=0.48\columnwidth]{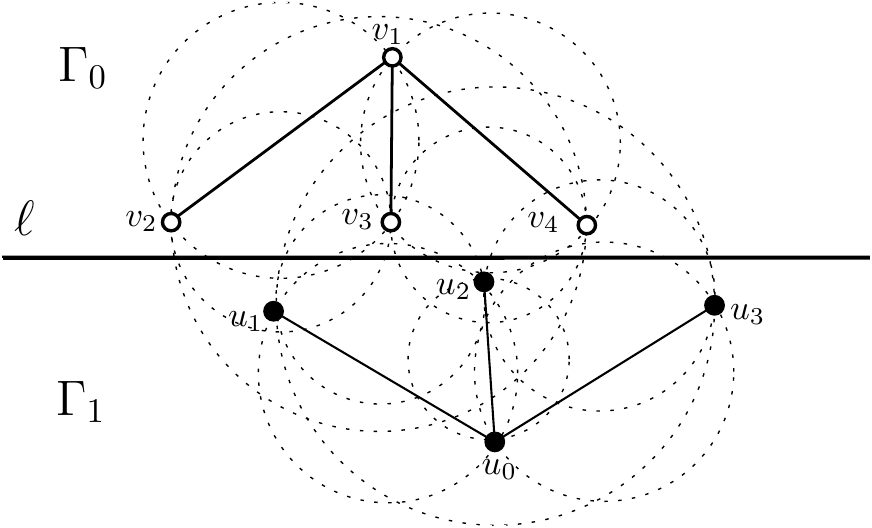}}
\hfill
		\subfigure[]{\label{fi:isolated-b}\includegraphics[width=0.48\columnwidth]{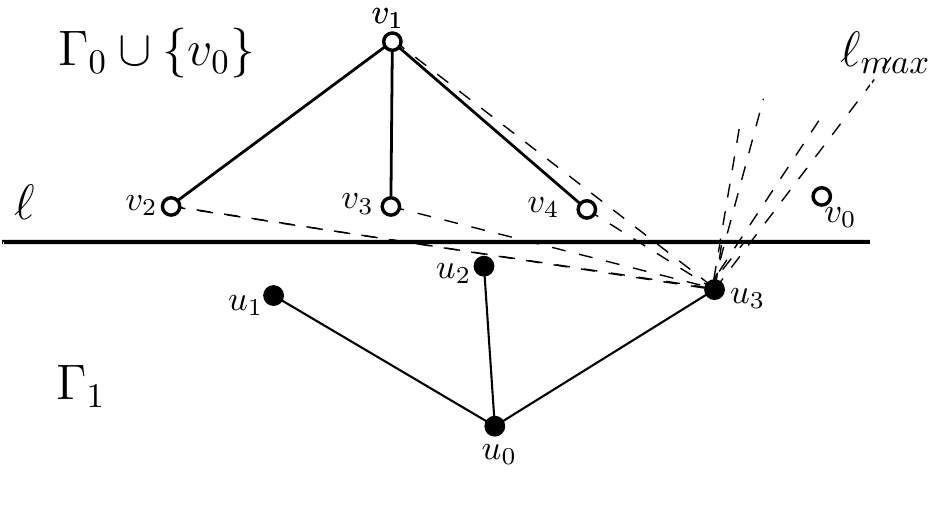}}
		\caption{(a) A linearly separable MWG-drawing $\langle \Gamma_0,\Gamma_1 \rangle $. (b)Adding an isolated vertex $v_0$ to $\Gamma_0$.}\label{fi:isolated}
	\end{figure}

Figure~\ref{fi:isolated-b} shows the addition of the isolated vertex $v_0$ to the MWG-drawing of Figure~\ref{fi:isolated-a}.
In fact, an isolated vertex can be added to the left (right) of whichever of the two drawings has the leftmost (rightmost) vertex, as long
as it is positioned sufficiently far enough to the left (right) of that vertex.
\cref{le:universal-vertex,le:isolated-vertex} are used in the following lemma, where we use colors to distinguish vertices in distinct partition sets.  

\begin{lemma}\label{le:2-stars}
 $\langle K_{1,n_0},K_{1,n_1} \rangle $ has a MWG-drawing if $|n_0 - n_1| \leq 2$ .
%
\end{lemma}

\begin{proof}
Observe that two independent sets whose sizes differ by at most 1, one consisting of red vertices and one consisting of blue vertices, admit a linearly separable MWG-drawing where the red vertices are above a horizonal separating line $\ell$ while  the blue  vertices are below $\ell$: start with one red vertex with coordinates $(0,1)$ and a blue vertex with coordinates $(-1,-1)$ and iteratively add red and blue vertices by applying the isolated vertex-addition procedure in the proof of \cref{le:isolated-vertex}.
Let $G_i = K_{1,n_i}$, for $i = 0, 1$.

Denote by $v_0$ the non-leaf vertex of $G_0$ and by $u_0$ the non-leaf vertex of $G_1$.
Assume first that $|n_0 - n_1| \leq 1$. By the previous observation, $\langle G_0 \setminus{\{v_0\}},G_1\setminus{\{u_0\}} \rangle $ admits a linearly separable MWG-drawing. Therefore, by \cref{le:universal-vertex} applied to $\langle G_0 \setminus{\{v_0\}},G_1\setminus{\{u_0\}} \rangle $ we have that $\langle G_0,G_1\setminus{\{u_0\}} \rangle $ admits a a linearly separable MWG-drawing. By \cref{le:universal-vertex} applied to $\langle G_0,G_1\setminus{\{u_0\}} \rangle $ we have that if $|n_0 - n_1| \leq 1$, the pair  $\langle G_0,G_1 \rangle $ is  MWG-drawable.

Consider now the case $|n_0 - n_1| = 2$ and assume that $n_0 > n_1$ (the proof when $n_1 > n_0$ is analogous). Let $v_1$ be a leaf of $G_0$. With the same reasoning as in the previous case, $\langle G_0 \setminus{\{v_0,v_1\}},G_1\setminus{\{u_0\}} \rangle $ admits a linearly separable MWG-drawing that we denote as $\langle \Gamma_0 \setminus{\{v_0,v_1\}},\Gamma_1\setminus{\{u_0\}} \rangle $. By the technique in the proof of \cref{le:universal-vertex}, we add the universal vertex $u_0$ to $\Gamma_1\setminus{\{u_0\}}$ in such a way that $u_0$ is the rightmost vertex of the linearly separable MWG-drawing $\langle \Gamma_0 \setminus{\{v_0,v_1\}},\Gamma_1 \rangle $. We now exploit the construction of \cref{le:isolated-vertex} to add the isolated vertex $v_1$ to  $\langle \Gamma_0 \setminus{\{v_0,v_1\}},\Gamma_1 \rangle $ and obtain a linearly separable MWG-drawing $\langle \Gamma_0 \setminus{\{v_0\}},\Gamma_1 \rangle $. Finally, we use \cref{le:universal-vertex} to construct an MWG-drawing of  $\langle G_0,G_1 \rangle $ also when $|n_0 - n_1| = 2$.
\end{proof}


\begin{lemma}\label{le:no-edge-blocking-vertex}
Let $\Gamma$ be a WG-drawing of a graph such that $\Gamma$ is linearly separable from its witness set $P$. If $\overline{u v}$ is an edge of $\Gamma$ and $z \in S_F[u, v]$ is a vertex of $\Gamma$ , then both $\overline{u z}$ and $\overline{v z}$ are edges of $\Gamma$.
\end{lemma}

\begin{proof}
Consider, w.l.o.g., the segment $\overline{u z}$. If it is not an edge of $\Gamma$, then it must have a witness in $S_N(u, z)$. But any such point will also be in $D[u, v]$, contradicting the fact that $\overline{u v}$ is an edge of $\Gamma$.
\end{proof}

\begin{restatable}[*]{lemma}{leverticalstrip}\label{le:vertical-strip}
Let $\Gamma$ be a WG-drawing of a graph such that $\Gamma$ is linearly separable from its witness set $P$. Let $u_0,u_1,v_0,v_1$ be such that $u_0, v_0, u_1, v_1$ induce a $C_4$ in $\Gamma$.
 We have that: (i) $v_0$ and $v_1$ are in opposite half-planes with respect to the line through $u_0,u_1$, and (ii) one of $\{v_0,v_1\}$ is a point of $S_N(u_0,u_1)$ and the other is {\em not} in $S[u_0,u_1]$.
\end{restatable}

%

By means of \cref{le:vertical-strip} we can restrict the set of complete bipartite graph pairs that are MWG-drawable.

\begin{restatable}[*]{lemma}{lenoktwothree}\label{le:no-k23}
 Let $\Gamma$ be a WG-drawing of a complete bipartite graph such that $\Gamma$ is linearly separable from its witness set. Then $\Gamma$ does not have $K_{2,3}$ as a subgraph.
\end{restatable}


We now characterize the MWG-drawable pairs of complete bipartite graphs. We recall that Aronov et al. prove that every complete bipartite graph admits a WG-drawing (Theorem 5 of ~\cite{DBLP:journals/comgeo/AronovDH13}). The following theorem can be regarded as an analog of the result by Aronov et al. in the context of MWG-drawings.

\begin{figure}[t]
	\centering
		\includegraphics[width=0.55\textwidth, page=1]{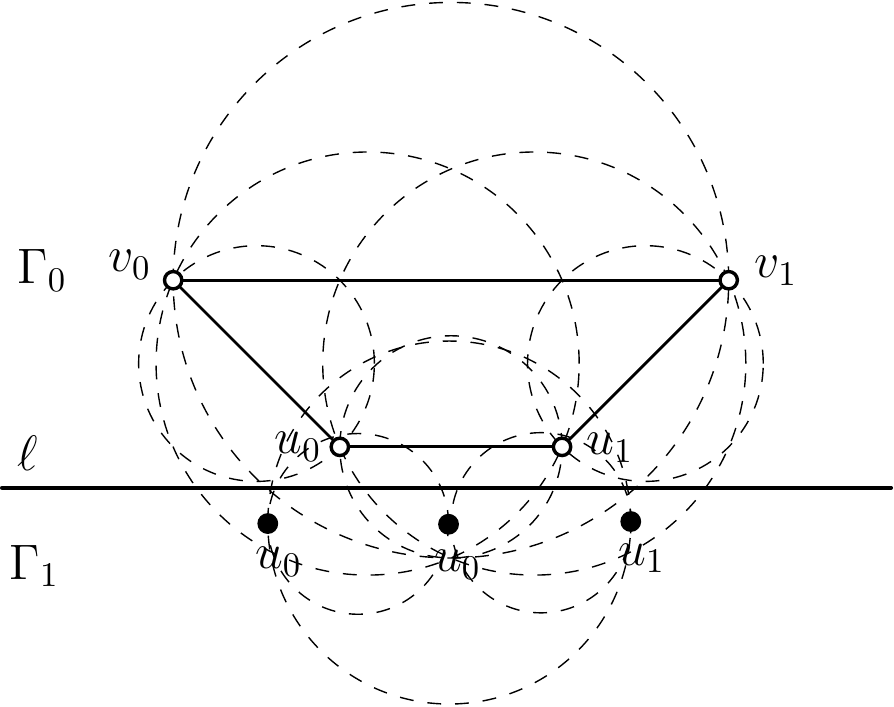}
			\caption{An MWG-drawing of $K_{2,2}$ and of an independent set of size three.}
	\label{fi:independent}
\end{figure}

\begin{theorem}\label{th:bipartite-characterization}
Let $\langle G_0,G_1 \rangle $ be a pair of complete bipartite graphs such that $G_i$ has $n_i$ vertices. The pair $\langle G_0,G_1 \rangle $ admits an MWG-drawing if and only if, for $i=0,1$, $G_i$ is either $K_{1,n_i - 1}$ or $K_{2,2}$ and $|n_0 - n_1| \leq 2$.
\end{theorem}
\begin{proof}
By Theorem~\ref{th:linear-separability}, any MWG-drawing $\langle \Gamma_0 , \Gamma_1 \rangle$ of $G_0$ and $G_1$ is linearly separable, so any witnes for a non-edge $\overline{u v}$ in $\Gamma_i$ must lie in $S(u,v)$. By Corollary~\ref{co:planarity} both $G_0$ and $G_1$ must be planar. By \cref{le:no-k23} each of the two graphs is either $K_{2,2}$ or a star (i.e. $K_{1,n_i - 1}$, $i=0,1$).
Together, these imply that the difference in the cardinalities of the vertex sets in the two graphs is at most two.

If $G_0 = K_{1,n_0 - 1}$ and $G_1 = K_{1,n_1 - 1}$, the theorem follows by \cref{le:2-stars}. If $G_0 = K_{2,2}$ and $G_1 = K_{2,2}$ the pair $\langle G_0,G_1 \rangle $ has an MWG-drawing as shown, for example, in Figure~\ref{fi:universal-a}. By removing one of the bottom-most vertices of $\Gamma_1$ in Figure~\ref{fi:universal-a} we obtain an MWG-drawing of $\langle K_{2,2},K_{1,2} \rangle $ and by removing both the bottom-most vertices of $\Gamma_1$ in Figure~\ref{fi:universal-a} we obtain an MWG-drawing of $\langle K_{2,2},K_{1,1} \rangle $. To complete the proof we have to show that $\langle K_{2,2},K_{1,3} \rangle $, $\langle K_{2,2},K_{1,4} \rangle $, and $\langle K_{2,2},K_{1,5} \rangle $ are also MWG-drawable pairs. To this end refer to Figure~\ref{fi:independent} that shows an MWG-drawing $\langle \Gamma_0,\Gamma_1 \rangle $ where $\Gamma_0$ is $K_{2,2}$ while $\Gamma_1$ is an independent set consisting of three vertices. By applying \cref{le:universal-vertex} we can add a universal vertex to $\Gamma_1$, thus obtaining an MWG-drawing of $\langle K_{2,2},K_{1,3} \rangle $. In order to construct MWG-drawings of $\langle K_{2,2},K_{1,4} \rangle $ and of $\langle K_{2,2},K_{1,5} \rangle $, notice that in Figure~\ref{fi:independent} $v_0$ is the leftmost vertex and $v_1$ is the rightmost vertex of $\langle \Gamma_0,\Gamma_1 \rangle $. By \cref{le:isolated-vertex} we can add either one isolated vertex or two isolated vertices to $\Gamma_1$. In the former case we obtain an MWG-drawing of $K_{2,2}$ and of an independent set of size four which can be extended to an MWG-drawing of $\langle K_{2,2},K_{1,4} \rangle $ by means of \cref{le:universal-vertex}. In the latter case, we again use \cref{le:universal-vertex} to add a universal vertex to the drawing of the independent set of size five and obtain an MWG-drawing of $\langle K_{2,2},K_{1,5} \rangle $.
\end{proof}

The following theorem is a consequence of \cref{th:bipartite-characterization} and of the constructive arguments of \cref{le:2-stars}



\begin{theorem}\label{th:bipartite-testing}
Let $\langle G_0,G_1 \rangle $ be a pair of complete bipartite graphs such that $G_0$ has $n_0$ vertices and $G_1$ has $n_1$ vertices. There exists an $O(n_0+n_1)$-time algorithm that tests whether  $\langle G_0,G_1 \rangle $ admits an MWG-drawing. In the  affirmative case, there exists an $O(n_0+n_1)$-time algorithms to compute an MWG-drawing of $\langle G_0,G_1 \rangle $ in the real RAM model of computation.
\end{theorem}

\section{MWG-drawable Complete $k$-partite graphs}\label{se:k-partite}

Aronov et al. also showed that there exists a complete multipartite graph, namely $K_{3,3,3,3}$, which does not admit a  WG-drawing (Theorem
15 of ~\cite{DBLP:journals/comgeo/AronovDH13}). We extend this result in the context of MWG-drawings by proving the following result.

\begin{figure}[t]
	\centering
		\includegraphics[width=0.7\textwidth, page=1]{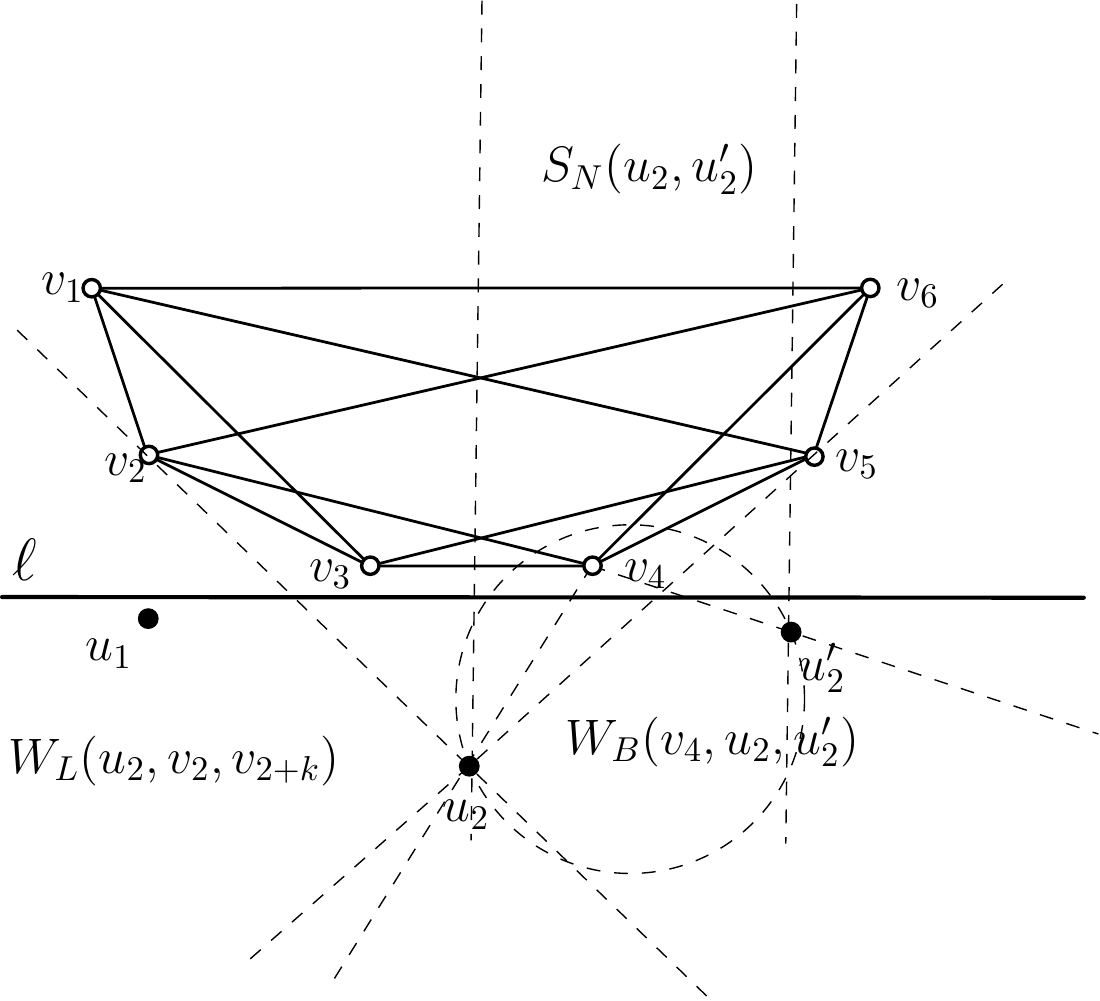}
			\caption{$v_4 \in S_N(u_2,u'_2)$ implies $ W_L(u_2,v_2,v_{2+k}) \cap W_B(v_4,u_2,u_2') = \emptyset$ .}
	\label{fi:nok222-k222}
\end{figure}

\begin{restatable}[*]{theorem}{thmultipartite}\label{th:no-k222-k222}
 Let $\langle G_0,G_1 \rangle $ be a pair of complete multi-partite graphs such that for each of the graphs every partition set has size at least two. The pair is mutually Gabriel drawable if and only if it is $\langle K_{2,2},K_{2,2} \rangle $.
\end{restatable}


\begin{proof} [Sketch] Let $\Gamma$ be a WG-drawing of a complete $k$-partite graph, with $k\geq 2$ such that $\Gamma$ is linearly separable from its witness set $P$ by a separating line $\ell$. Note that any induced subgraph $G'$ of $G$ admits a WG-drawing with witness set $P$, which can be derived from $\Gamma$ by removing the vertices not in $G'$.  By this observation, \cref{th:linear-separability}, and \cref{le:no-k23} we conclude that if $\langle G_0,G_1 \rangle $ is a pair of complete MWG-drawable multi-partite graphs, then neither $G_0$ nor $G_1$ can have $K_{2,3}$ as a subgraph. Therefore we can assume that all partition sets in each of the two graphs have size exactly two. Refer to Figure~\ref{fi:nok222-k222}.

The proof proceeds by first showing that $\mathit{CH}(\Gamma)$ is a {\em convex terrain} with respect to $\ell$; that is  for each vertex $v$ on the boundary of $\mathit{CH}(\Gamma)$, the segment from $v$ to $\ell$ perpendicular to $\ell$ does not intersect $\mathit{CH}(\Gamma)$. Using this property we can order the vertices of $\Gamma$ by increasing $x$-coordinate and show that the $i$-th partite set consists of vertices $v_i$, $v_{i+k}$.  Let $p_i$ be a witness for vertices $v_i$ and $v_{k+i}$ and let $p_j$ be a witness for vertices $v_j$ and $v_{k+j}$. Thirdly, we show that if $i < j$ then $p_i \in W_L(p_j,v_j,v_{j+k})$; if $i > j$ then $p_i \in W_R(p_j,v_j,v_{j+k})$. Consider now an MWG-drawing $\langle \Gamma_0,\Gamma_1 \rangle $ where $\Gamma_0$ is in the upper half-plane with respect to the separating line $\ell$ and it has at least three partition sets. Let $u_1$, $u_2$ and $u_3$ be three vertices of $\Gamma_1$ that act as witnesses for partition sets $\{v_1,v_{1+k}\}$, $\{v_2,v_{2+k}\}$, and $\{v_3,v_{3+k}\}$, respectively. We have $u_1 \in W_L(u_2,v_2,v_{2+k})$ and  $u_3 \in W_R(u_2,v_2,v_{2+k})$. Let $u_2'$ be the vertex of $\Gamma_1$ that is in the same partition set as $u_2$ and assume that $u_2'$ is to the right of $u_2$ (the proof in the other case being symmetric). Let $v$ be a vertex of $\Gamma_0$ that is a witness of $u_2$ and $u_2'$ (i.e. $v \in D[u_2, u'_2]$).

By Property~\ref{pr:common-neighbor} either $v \in \{v_2, v_{2+k}\}$ or $v \in W_T(u_2, v_2, v_{2+k})$, hence $v \in W_T[u_2, v_2, v_{2+k}]$.
Again by Property~\ref{pr:common-neighbor}, all vertices of $\Gamma_1$ must lie in $W_B[v,u_2,u_2']$.
Because $v \in W_T[u_2, v_2, v_{2+k}]$,  $W_B[v,u_2,u_2']$ is disjoint from $W_L(u_2,v_2,v_{2+k})$.
But $u_1 \in W_L(u_2,v_2,v_{2+k})$, a contradiction.
\end{proof}

\section{Open Problems}\label{se:open}

 The results of this paper naturally suggest many interesting  open problems. For example: (i) Can one give a complete characterization of those pairs of complete multipartite graphs that admit an MWG-drawing extending \cref{th:no-k222-k222} by taking into account graphs some of whose  partition sets have size one? It is not hard to see that the ideas of \cref{le:universal-vertex,le:2-stars,} can be used to construct MWG-drawings of graph  pairs of the form $\langle K_{1,\cdots,1,n_0},K_{1,\cdots,1,n_1} \rangle $ as long as the number of partition sets of size one in the two graphs differ by at most two. However, this may not be a complete characterization. (ii) Which other pairs of diameter-2 graphs admit an MWG-drawing? (iii) Which pairs of (not necessarily complete) bipartite graphs admit an MWG-drawing?
 (iv) Finally, it would be interesting to study mutual witness drawings for other proximity regions.

\bibliographystyle{abbrvurl}
\bibliography{biblio}

\newpage

\newpage

\appendix
\makeatletter
\noindent
\rlap{\color[rgb]{0.51,0.50,0.52}\vrule\@width\textwidth\@height1\p@}%
\hspace*{7mm}\fboxsep1.5mm\colorbox[rgb]{1,1,1}{\raisebox{-0.4ex}{%
		\large\selectfont\sffamily\bfseries Appendix}}%
\makeatother

\section*{Proofs omitted from Section~\ref{se:bipartite}}

\coplanarity*
\begin{proof}
Let $\langle \Gamma_0,\Gamma_1 \rangle $ be an MWG-drawing of $\langle G_0,G_1 \rangle $. By \cref{th:linear-separability}, $\Gamma_0$ and $\Gamma_1$ are linearly separable. Any two vertex-disjoint edges of $\Gamma_i$ induce an alternating 4-cycle and thus, by \cref{le:alternating-4-cycle}, they cannot cross. And, of course, two edges that share a vertex cannot cross, so both graphs are planar.
\end{proof}

\coconvexity*
\begin{proof}
By the same argument as in \cref{co:planarity}, no two edges of $\Gamma$ can cross. Let $C$ be a 4-cycle formed ny four edges $\overline{u_0v_0}, \overline{v_0u_1}, \overline{u_1v_1}, \overline{v_1u_0}$  of $\Gamma$. Since any the two pairs of non adjacent edges of $C$ form an alternating 4-cyle with the two non-edges $\overline{u_0u_1}$ and $\overline{v_0v_1}$, by \cref{le:alternating-4-cycle} we have that  non-edges $\overline{u_0u_1}$ and $\overline{v_0v_1}$ cross in $\Gamma$, which implies that $C$ is a convex polygon.
\end{proof}

\leuniversalvertex*
\begin{proof}
Let $\langle \Gamma_0,\Gamma_1 \rangle $ be a linearly separable MWG-drawing of $\langle G_0,G_1 \rangle $  and let $\ell$ be the horizontal line separating $\Gamma_0$ from $\Gamma_1$; see, for example, Figure~\ref{fi:universal-a}. Assume that $\Gamma_0 $ is in the top half-plane defined by $\ell$ and let $\ell'$ be a vertical line not passing through any vertices of $\Gamma_0$. Let $p$ be the intersection point of $\ell$ with $\ell'$. For each vertex $v \in \Gamma_0$ let  $\ell_v$ be the line passing through $v$ and orthogonal to segment $\overline{pv}$. Let $p_v$ be the intersection point of $\ell_v$ with $\ell'$ and let $p'$ the topmost of such intersection points defined by considering all vertices of $\Gamma_0$. Also, let $p''$ be the topmost intersection point between $D[u, u']$ and $\ell'$, taken over all edges $\overline{u u'}$ of $\Gamma_1$.  Let $p_{\max}$ be any point of $\ell'$ above both $p'$ and $p''$. A linearly separable MWG-drawing of $\langle G_0 \cup \{v_0\},G_1 \rangle $ such that $v_0$ is a universal vertex of $G_0$ can be obtained from $\langle \Gamma_0,\Gamma_1 \rangle $ by adding $v_0$ at any point along $\ell'$ above $p_{\max}$. See, for example, Figure~\ref{fi:universal-b}. By construction, we have that: (i) for any vertex $v \in \Gamma_0$, we have $D[v, v_0] \cap \ell = \emptyset$ which implies that $v_0$ is adjacent to all vertices of $\Gamma_0$; (ii) for any pair $u u'$ of vertices of $\Gamma_1$, the Gabriel disk $D[u, u']$ contains a vertex of $\Gamma_0 \cup \{v_0\}$ if and only if it contains a vertex of $\Gamma_0$; (iii) since no vertex has been added or deleted in $\Gamma_1$, all edges and non-edges of $\Gamma_0$ are maintained in $\langle \Gamma_0 \cup \{v_0\},\Gamma_1 \rangle $.
\end{proof}


\leisolatedvertex*
\begin{proof}
  Let $\langle \Gamma_0,\Gamma_1 \rangle $ be a linearly separable MWG-drawing of $\langle G_0,G_1 \rangle $, let $\ell$ be the horizontal line separating $\Gamma_0$ from $\Gamma_1$, and assume that $\Gamma_0 $ is in the top half-plane defined by $\ell$. Let $u$ be the rightmost vertex of the MWG-drawing. If $u$ belongs to $\Gamma_1$ we add an isolated vertex to $\Gamma_0$, else we add an isolated vertex to $\Gamma_1$. For concreteness, we describe the construction in the case that $u$ is a vertex of $\Gamma_1$, the other case being symmetric. See, for example, Figure~\ref{fi:isolated-a}. For each vertex $v \in \Gamma_0$ let  $\ell_v$ be the line through $u$ orthogonal to segment $\overline{v u}$. Let $p_v$ be the intersection point of $\ell_v$ with $\ell$ and let $p_{\max}$ be the rightmost of such intersection points defined by considering all vertices of $\Gamma_0$ and $\ell_{\max}$ be the line from the set $\{\ell_v : v \in \Gamma_0\}$ intersecting $\ell$ at $p_{\max}$. See, for example, Figure~\ref{fi:isolated-b}. Add a vertex $v_0$ to the right of $\ell_{\max}$ and above $\ell$. We now show that $\langle \Gamma_0 \cup \{v_0\},\Gamma_1 \rangle $ is an MWG-drawing.

  Since $v_0$ is to the right of any vertex of $\Gamma_1$, $v_0$ is in the top half-plane and $\Gamma_1$ is in the bottom half-plane, by Property~\ref{pr:strip} we have that for any pair of vertices $u_i,u_j$ of $\Gamma_1$, $v_0$ is not a point of $D[u_i, u_j]$.
It follows that for any pair
$u_i,u_j$ of vertices of $\Gamma_1$, the Gabriel disk $D[u_i, u_j]$ contains a vertex of $\Gamma_0 \cup \{v_0\}$ if and only if it contains a vertex of $\Gamma_0$. Also, for any vertex $v \in \Gamma_0$, we have $\angle v u v_0 > \frac{\pi}{2}$ which implies that $u \in D[v, v_0]$ and $v_0$ is not adjacent to any vertex of $\Gamma_0$. Finally, since no vertex has been added or deleted in $\Gamma_1$, all edges and non-edges of $\Gamma_0$ are maintained in $\langle \Gamma_0 \cup \{v_0\},\Gamma_1 \rangle $. We conclude that if $\langle \Gamma_0,\Gamma_1 \rangle $ is a linearly separable MWG-drawing, also $\langle \Gamma_0 \cup \{v_0\},\Gamma_1 \rangle $ is a linearly separable MWG-drawing.
\end{proof}

The following property can be proved with elementary geometric arguments.

\begin{property}\label{pr:strip}
Let $\langle \Gamma_0,\Gamma_1 \rangle $ be a linearly separable MWG-drawing of $\langle G_0,G_1 \rangle $, let $\ell$ be the (horizontal) line separating $\Gamma_0$ from $\Gamma_1$, let $i \in \{0,1\}$, let $u_0,u_1$ be a pair of non-adjacent vertices of $\Gamma_i$ with $x(u_0) < x(u_1)$, and let $z$ be a vertex of $\Gamma_{1-i}$ such that $z \in D[u_0, u_1]$. We have that $z \in S_N(u_0,u_1)$.
\end{property}

\leverticalstrip*
\begin{proof}
 Let $\lambda$ be the line through $u_0$ and $u_1$. If both $v_0$ and $v_1$ lie in the same half-plane defined by $\lambda$ we have that one of $\overline{v_0u_0}, \overline{v_0u_1}$ crosses one of $\overline{v_1u_0}, \overline{v_1u_1}$, which is impossible by \cref{le:alternating-4-cycle}.  Hence, Property (i) holds. Let $p$ be a witness point such that $p \in D[u_0, u_1]$; by Property~\ref{pr:strip} $p \in S_N(u_0,u_1)$. Let $v_0$ be the vertex in the same half-plane as $p$ with respect to $\lambda$. Vertex $v_0$ must also be a point of  $S_N(u_0,u_1)$ or else we would have that either $\angle u_0pv_0 > \frac{\pi}{2}$ or $\angle u_1pv_0 > \frac{\pi}{2}$, which is impossible because both $\overline{v_0u_0}$ and $\overline{v_0u_1}$ are edges of $\Gamma$.

Now suppose $v_1 \in S[u_0,u_1]$. Then either $v_1 \in S[u_0,v_0]$ or $v_1 \in S[v_0,u_1]$.
Assume that $v_1 \in S[u_0,v_0]$ (the other case is analogous). Then, by \cref{le:no-edge-blocking-vertex},
$\overline{v_0 v_1}$ is an edge of $\Gamma_0$, a contradiction. Thus $v_1 \not \in S[u_0,u_1]$
and hence Property (ii) holds.
\end{proof}

\lenoktwothree*
\begin{proof} Suppose $\Gamma$ contains $K_{2,3}$ as a subgraph. Let $u_0,u_1,u_2, v_0, v_1$ be the vertices of $K_{2,3}$ such that  $u_0,u_1,u_2$ are in the same partition set and assume $x(u_0) \leq x(u_1) \leq x(u_2)$. Consider the 4-cycle $C_0$ with vertices $v_0, u_0, v_1, u_1$ and the 4-cycle $C_1$ with vertices $v_0, u_1, v_1, u_2$. By \cref{co:4-cycle-convexity} $\Gamma$ is such that both $C_0$ and $C_1$ are convex quadrilaterals, which implies that $v_0$, $u_1$, and $v_1$ are vertically aligned points. However, by \cref{le:vertical-strip} one of $\{v_0,v_1\}$--say $v_0$--must be a point of $S_N(u_0,u_1)$ that is $x(u_0) < x(v_0) < x(u_1)$, a contradiction.
 \end{proof}

\section*{Proofs omitted from Section~\ref{se:k-partite}}

Observe that if $\Gamma$ is a WG-drawing of a graph $G$ with witness set $P$, then any induced subgraph $G'$ of $G$ admits a WG-drawing with witness set $P$, which can be derived from $\Gamma$ by removing the vertices not in $G'$. Also, by \cref{th:linear-separability}, if $G_0$ and $G_1$ are MWG-drawable, their drawings must be linearly separable. By this observation and \cref{le:no-k23} we conclude that if $\langle G_0,G_1 \rangle $ is a pair of complete MWG-drawable multi-partite graphs, then neither $G_0$ nor $G_1$ can have $K_{2,3}$ as a subgraph. It follows that in order to prove \cref{th:no-k222-k222} we can assume that all partition sets in each of the two graphs have size exactly two.

To prove \cref{th:no-k222-k222} we introduce some terminology and a technical result from Aronov et al~\cite{DBLP:journals/comgeo/AronovDH13}. Let $\ell$ be a  line and let $\Pi$ be a convex polygon in one of the half-planes of $\ell$. We say that $\Pi$ is a \emph{convex terrain with respect to $\ell$} if for each vertex $v$ on the boundary of $\Pi$, the segment from $v$ to $\ell$ perpendicular to $\ell$ does not intersect $\Pi$. To be consistent with~\cite{DBLP:journals/comgeo/AronovDH13}, we associate each partition set of $\Gamma$ with a distinct color. We shall also sometimes say that two vertices are in the same color class to mean that they belong to the same partition set.

\begin{lemma}~\cite{DBLP:journals/comgeo/AronovDH13}\label{le:convex-position}
Let $\Gamma$ be a WG-drawing of a complete $k$-partite graph, with $k\geq 2$. Then $\mathit{CH}(\Gamma)$ contains at least two vertices of each color class on its boundary.
\end{lemma}

\begin{figure}[t]
	\centering
		\includegraphics[width=0.6\textwidth, page=1]{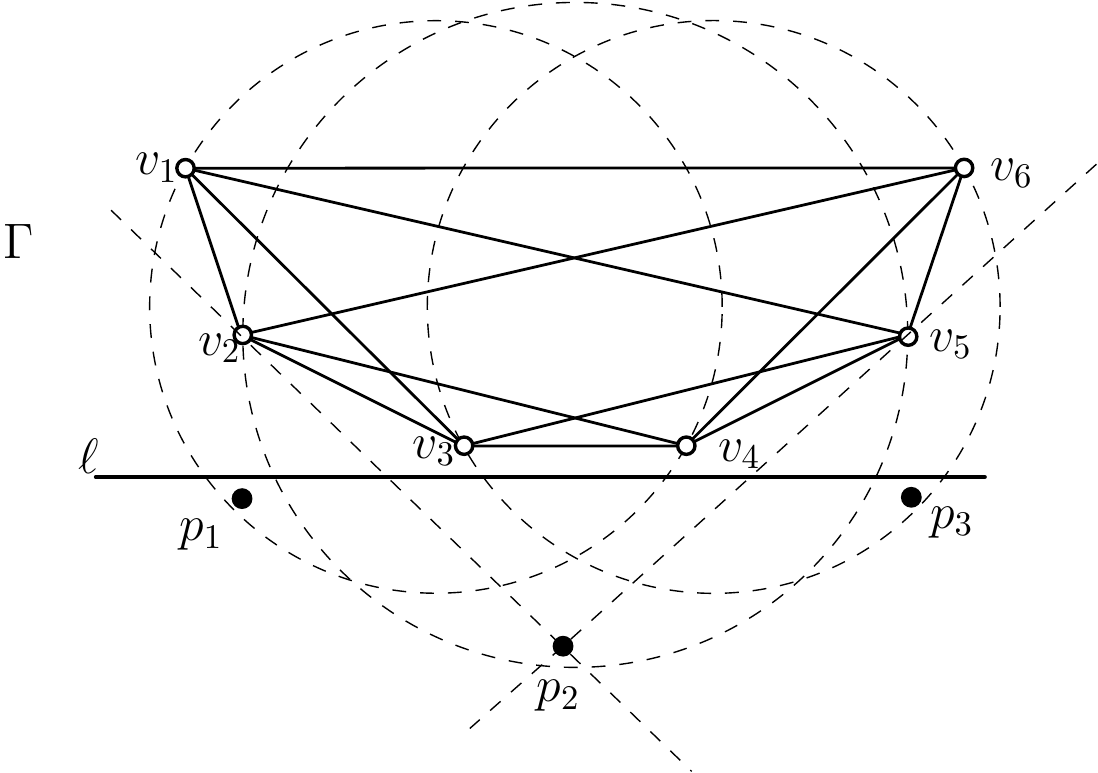}
			\caption{$p_1 \in W_L(p_2,v_2,v_5)$ and $p_3 \in W_R(p_2,v_2,v_5)$.}
	\label{fi:terrain}
\end{figure}

\begin{lemma}\label{le:convex-terrain}
Let $\Gamma$ be a WG-drawing of a complete $k$-partite graph, with $k\geq 2$ such that $\Gamma$ is linearly separable from its witness set by a separating line $\ell$ and such that every color class of $\Gamma$ has size two. Then $\mathit{CH}(\Gamma)$ is a convex terrain with respect to $\ell$.
\end{lemma}
\begin{proof}
As usual, we assume that $\ell$ is a horizontal line. Note that by Lemma~\ref{le:convex-position}, every vertex of $\Gamma$ is on the boundary of  $\mathit{CH}(\Gamma)$. Label the vertices as $v_1, \ldots, v_{2 k}$ in order of increasing $x$-coordinate (we can assume that the vertices have distinct $x$-coordinates).
We first show that, for $i = 1, \ldots, k$, $v_i$ and $v_{k+i}$ have the same color.

Let $u_1 = v_1$ and let $u'_1$ be the other vertex having the same color as $u_1$; call this color $c_1$. Now let  $w_1$ and  $w'_1$ be two vertices of the same color (but different from $c_1$). By Lemma~\ref{le:vertical-strip}, exactly one of $w_1,  w'_1$ is in $S(u_1,  u'_1)$; in fact it is in $S_N(u_1,  u'_1)$. Thus, $S_N(u_1,  u'_1 )$ contains exactly one vertex of every color other than $c_1$.
Thus $u'_1 = v_{k+1}$ and $v_1, \ldots , v_k$ all have distinct colors $c_1, \ldots, c_k$.

Now let $u_2 = v_2$ and let $u'_2$ be the other vertex having the same color as $u_2$; so $u_2$ has color $c_2$. By the same argument,
$S(u_2,  u'_2 )$ (in fact $S_N(u_2,  u'_2 )$) contains exactly one vertex of each of the other colors and so $u'_2$ must be $v_{k+2}$, and, similarly, for $i = 1, \ldots k$,
we have that $v_i = v_{k+i}$, and that $v_{i+1}, \ldots, v_{k+i-1}$ are in  $S_N(v_i, v_{k+i})$.

Thus $v_1, \ldots, v_{2 k}, v_1$ is a convex polygon and so every diagonal $v_i, v_j$ is contained within $\mathit{CH}(\Gamma)$ and thus, for each $v_i$, the vertical segment from $v_i$ to $\ell$ does not intersect $\mathit{CH}(\Gamma)$, making $\mathit{CH}(\Gamma)$ a convex terrain with respect to $\ell$.
\end{proof}

An example of the configuration described in the statement of \cref{le:convex-terrain} is shown in Figure~\ref{fi:terrain}.
In the following lemma we assume, without loss of generality, that for every $p \in P$, $p$ is contained in the Gabriel disk of some pair of vertices of $\Gamma$. Indeed, if $P$ contained a point $p$ which is not in any Gabriel disk of any pair of non-adjacent vertices of $\Gamma$, then $P \setminus{\{p\}}$ would also be a witness set for $\Gamma$.

\begin{lemma}\label{le:witness-ordering}
Let $\Gamma$ be a WG-drawing of a complete $k$-partite graph, with $k\geq 2$ such that $\Gamma$ is linearly separable from its witness set $P$ by a separating line $\ell$ and such that every color class of $\Gamma$ has size two. Let $p_i$ be a witness for vertices $v_i$ and $v_{k+i}$ and let $p_j$ be a witness for vertices $v_j$ and $v_{k+j}$. If $i < j$ then $p_i \in W_L(p_j,v_j,v_{j+k})$; if $i > j$ then $p_i \in W_R(p_j,v_j,v_{j+k})$.
\end{lemma}
\begin{proof}
We assume, w.l.o.g, the separating line is horizontal and that $\Gamma$ is in the upper half-plane (see for example Figure~\ref{fi:terrain}).
We begin with an observation. Let $\{v_i, v_{i+k}\}$ be one of the color classes of $\Gamma$ and let $p_i \in P$ be in $D[v_i, v_{i+k}]$. Since every other vertex $v$ of $\Gamma$ is adjacent to both $v_i$ and $v_{i+k}$, by Property~\ref{pr:common-neighbor} it follows that $v$ must lie in $W_T(p_i, v_i, v_{i+k})$ and that no $p \in P$ can be a witness for two different color classes.

It is also the case that if $\{v_i, v_{i+k}\}$ is a color class distinct from $\{v_j, v_{j+k}\}$, then no witness $p_i$ for $\{v_i, v_{i+k}\}$ can be contained in $W_T(p_j, v_j, v_{j+k})$, because then $p_i$ would be contained in $\triangle(p_j, v_j, v_{j+k})$, which would imply that
\[\angle v_j,p_i,v_{j+k} > \angle v_j,p_j,v_{j+k} \geq \frac{\pi}{2},\]
which makes $p_i$ a witness for $\{v_j, v_{j+k}\}$ as well as for $\{v_i, v_{i+k}\}$.

Similarly, $p_i$ cannot be contained in $W_B(p_j, v_j, v_{j+k})$ since then, in order for $W_T(p_i, v_i, v_{i+k})$ to contain both of $\{v_j, v_{j+k}\}$, $W_T(p_i, v_i, v_{i+k})$ would also need to contain $p_j$, which we just showed cannot occur.
So, either $p_i \in W_L(p_j, v_j, v_{j+k})$ or $p_i \in W_R(p_j, v_j, v_{j+k})$.

Assume that $i < j$ (the proof when $i > j$ is analogous). Since, by \cref{le:convex-terrain}, the vertices of  $\Gamma$ form a convex terrain with respect to the separating line, we have that $p_i$ is in $W_L(p_j,v_j,v_{j+k})$, as otherwise $W_T(p_i, v_i, v_{i+k})$ would not contain $v_j$.
\end{proof}

We are now ready to prove \cref{th:no-k222-k222}. As noted at the beginning of this section, it suffices to consider graphs all of whose partitions sets have size two.

\thmultipartite*
\begin{proof} Assume that $\Gamma_0$ is above the separating line and that it has at least three color classes. Let $u_1$, $u_2$ and $u_3$ be three vertices of $\Gamma_1$ that act as witnesses for color classes $\{v_1,v_{1+k}\}$, $\{v_2,v_{2+k}\}$, and $\{v_3,v_{3+k}\}$, respectively. By \cref{le:witness-ordering}, we have $u_1 \in W_L(u_2,v_2,v_{2+k})$ and  $u_3 \in W_R(u_2,v_2,v_{2+k})$. Let $u_2'$ be the vertex of $\Gamma_1$ that is in the same color class as $u_2$ and assume that $u_2'$ is to the right of $u_2$ (the proof in the other case being symmetric). Let $v$ be a vertex of $\Gamma_0$ that is a witness of $u_2$ and $u_2'$ (i.e. $v \in D[u_2, u'_2]$). See for example Figure~\ref{fi:nok222-k222} where vertex $v$ is $v_4$. Vertex $v$ is above the separating line and, by Property~\ref{pr:strip}, $v \in S_N(u_2,u_2')$. Also either $v$ coincides with one of $\{v_2,v_{2+k}\}$ or, by Property~\ref{pr:common-neighbor} $v \in  W_T(u_2,v_2,v_{2+k})$. In either case, we have that $ W_L(u_2,v_2,v_{2+k}) \cap W_B(v,u_2,u_2') = \emptyset$. However, $W_B(v,u_2,u_2')$ must, by Property~\ref{pr:common-neighbor},  contain all vertices of $\Gamma_1$, a contradiction since $u_1 \in W_L(u_2,v_2,v_{2+k})$.
\end{proof}

\end{document}